\documentclass[12pt]{article}

\usepackage{amssymb, caption}
\usepackage{amsfonts}
\usepackage{graphicx}
\usepackage{amsmath}
\usepackage{endnotes}
\usepackage{epsfig}
\usepackage{rotating}
\usepackage{amsthm}

\usepackage{etex}
\newcommand{\ci}{\perp\!\!\!\perp}
\usepackage{auto-pst-pdf} 
\usepackage{pstricks}
\usepackage{pst-node, verbatim}
\usepackage{amsfonts, lscape}
\usepackage{amsmath}
\usepackage{amsmath}
\usepackage{amsfonts}
\usepackage{amssymb}
\usepackage{graphicx}
\usepackage{pstricks, pstricks-add, pgfplots, tikz}
\usepackage{pst-node, verbatim, textcomp}
\usepackage[miktex]{gnuplottex}

\usetikzlibrary{patterns}
\usepgfplotslibrary{external}
\pgfplotsset{compat = newest} 
\usepgfplotslibrary{groupplots}
\usetikzlibrary{matrix,calc}
\pgfplotsset{compat = newest}
\usepgfplotslibrary{fillbetween}
\usepgfplotslibrary{groupplots}
\usetikzlibrary{pgfplots.groupplots}
\usetikzlibrary{plotmarks}
\usetikzlibrary{patterns}
\usepgfplotslibrary{external}
\pgfplotsset{compat = newest} 
\input{texdraw}
\usepgfplotslibrary{external}
\usepgfplotslibrary{groupplots}

\newcounter{remark_counter}
\newtheorem{remark}[remark_counter]{Remark}

\renewcommand{\Pr}{\mathbb{P}}

\newcommand{\mR}{\mathcal{R}}
\newcommand{\mZ}{\mathcal{Z}}

\renewcommand{\Pr}{\mathbb{P}}
\newcommand{\Ep}{{\mathrm{E}}}

\newcommand{\sgn}{\text{sgn}}

\renewcommand{\mid}{\,|\,}

\usepackage{tikz}
\usepackage{pgfplots}
\usepgfplotslibrary{groupplots}
\usetikzlibrary{pgfplots.groupplots}
\usetikzlibrary{plotmarks}
\usetikzlibrary{patterns}
\usepgfplotslibrary{external}
\pgfplotsset{compat=newest} 
\newrgbcolor{lightblue}{0.9 0.9 1}
\newrgbcolor{lightred}{1 0.9 0.9}
\newrgbcolor{lightgreen}{0.9 1 0.9}
\pgfplotsset{every axis legend/.append style={
		at={(0,0)},
		anchor=north east}}
\usepackage{pgfplotstable}

\usepackage{amsthm}

\newcommand{\mW}{\mathcal{W}}
\newcommand{\mY}{\mathcal{Y}}

\makeatletter
\def\thmheadbrackets#1#2#3{%
  \thmname{#1}\thmnumber{\@ifnotempty{#1}{ }\@upn{#2}}%
  \thmnote{ {\the\thm@notefont[#3]}}}
\makeatother

\newtheoremstyle{brakets}
  {}
  {}
  {\itshape}
  {}
  {\bfseries}
  {.}
  { }
  {\thmheadbrackets{#1}{#2}{#3}}

\theoremstyle{brakets}
\newtheorem{theorem}{Theorem}
\newtheorem{lemma}{Lemma}

\newtheorem{algorithm}[theorem]{Algorithm}

\newtheorem{step}{Step}

\newtheorem{assumption}{Assumption}

\renewcommand{\baselinestretch}{1.3}

\begin{document}

\title{A simple distributional difference-in-differences estimator for univariate and bivariate outcomes\thanks{An earlier version of this manuscript was circulated under the title "Distribution Regression Difference in Differences".}}
\author{Iv\'an Fern\'andez-Val\thanks{Boston University} \and Aico van Vuuren\thanks{University of Groningen.} \and  Francis Vella\thanks{Georgetown University.}
\and  Jonas Meier}
\date{\today}

\maketitle

\begin{abstract}
We provide a simple distribution regression estimator for treatment effects in the  difference-in-differences (DiD) design. Our procedure is particularly useful when the treatment effect differs across the distribution of the outcome variable. Our proposed estimator easily incorporates covariates and, importantly, can be extended to settings where the treatment potentially affects the joint distribution of multiple outcomes. Our key identifying restriction is that the untreated outcome distribution does not exhibit an interaction effect of group and time. This assumption results in a parallel trend assumption on a transformation of the distribution.  We highlight the relationship between our procedure and assumptions with the changes-in-changes approach of Athey and Imbens (2006). We also reexamine the Card and Krueger (1994) study of the impact of minimum wages on employment to illustrate the utility of our approach.   

\end{abstract}

 
\newpage
\section{Introduction}
The remarkable popularity of the difference-in-difference (DiD) estimator to evaluate the impact of policy interventions, inspired by Snow (1885) and introduced to economic applications by David Card (see, for example, Card 1990, Card and Krueger, 1994), is one of the most striking features of empirical work on treatment and policy effects. While the methodological innovations in this literature (see Arkhangelsky and Imbens, 2024 for a recent review) include the use of constructed control groups, the staggered timing of treatments, and fuzzy rather than sharp designs, the vast majority of the associated empirical work has estimated the mean effect of the treatment on a single economic outcome. This seems somewhat limited and a fuller evaluation of a policy treatment would be based on an examination of the marginal and joint distributions of all outcomes it potentially influences. This paper provides a simple procedure for estimating distributional treatment effects in the presence of a single treatment when the outcomes of interest are potentially multivariate. 

An initial methodological innovation focusing on distributional effects in DiD estimation is the changes-in-changes procedure of Athey and Imbens (2006), which estimates the counterfactual distribution of the treated group in the absence of treatment and compares it to the observed distribution under treatment. Torous et al. (2024) extend the Athey and Imbens approach (2006) to the multivariate outcome setting. 
Other work has adapted DiD estimation to examine the treatment effects at different quantiles of the outcome via the use of quantile regression. This includes, for example, Callaway and Li (2018, 2019). In contrast, Dube (2019), Goodman-Bacon (2021), and Goodman-Bacon and Schmidt (2020) employ conventional DiD estimation to explore the impact of the treatment at different points of the outcome distribution. Other distributional approaches include Kim and Wooldridge (2023) and Biewen, Fitzenberger, and Rümmele (2022). The former proposes an inverse probability weighting based procedure, while the latter employs a distribution regression (DR) approach. In this paper we also adopt a DR approach to constructing counterfactuals. In contrast to Biewen, Fitzenberger, and Rümmele (2022), who construct the counterfactual distributions via linear probability models, we employ non-linear link functions such as probit or logit models. This has a number of advantages, which we discuss below. In addition, we provide the associated identifying conditions required for this form of implementation of DR-DiD.

While DiD has typically been employed to evaluate the treatment effect on a specified economic outcome, there are many instances in which the treatment may affect multiple outcomes.  For example, a change in tax rates on earnings of married couples may affect the hours of work of both husbands and wives. An analysis of such a tax change should include the impact on each of the outcomes. However, a richer analysis would not only examine the impact on the respective marginal hours distributions of husbands and wives but also the joint distribution of hours. Alternatively, while evaluations of minimum wage laws typically evaluate their impact on total employment, they may also affect the joint distribution of part-time and full-time employment. We illustrate how this joint effect can be evaluated via the bivariate distribution regression (BDR) approach of Chernozhukov et al. (2025). This requires that we first estimate the joint distribution of the bivariate potential outcome by BDR and then construct the appropriate counterfactual. The treatment effects are then obtained via the appropriate comparisons. A recent alternative to this approach is extending the changes-in-changes procedure to multiple outcomes as is done in Torous et al. (2024). 

We investigate the asymptotic properties of our estimators with panel data under cross-section independence, allowing for unrestricted time series  dependence. This is an extension of the asymptotic theory of counterfactuals of Chernozhukov et al. (2013) for the univariate case and  Chernozhukov et al. (2025) for the bivariate case. 

The following section introduces the model and provides an analysis of the univariate case without covariates. We also extend our analysis to include covariates and contrast our approach with the Athey and Imbens (2006) changes-in-changes procedure.  Section \ref{sec:multiple} extends our analysis to the multiple outcome case and Section \ref{sec:estimation} discusses estimation. Section \ref{sec:asymptotic} discusses the asymptotic theory regarding our estimators. Section \ref{sec:empirical} provides an empirical illustration via an application of our procedure to the data employed in the Card and Krueger (1994) study of the impact of increasing the minimum wage on employment. Section \ref{sec:conclusions} concludes.

\section{Econometric analysis of the univariate case}

\subsection{Model without covariates}

\label{ss:without_covariates}

Consider the canonical DiD design with 2 periods, $T \in \{0,1\}$, and 2 groups, $G \in \{0,1\}$ in which a binary treatment, $D \in \{0,1\}$, is administered only to the treatment group with $G = 1$ in the second period $T=1$. Let $Y^0$ and $Y^1$ denote the potential outcomes under the non-treated and treated statuses. The observed outcome is $Y = Y^0 (1-D) + Y^1 D$, which corresponds to $Y^0$ for both groups at $T=0$, $Y^0$ for $G=0$ at $T=1$, and $Y^1$  for $G=1$ at $T=1$. Note that this implicitly imposes a non-anticipation assumption as we do not distinguish between the outcomes of the treated and non-treated state for $G=1$ in period $T=0$.
 
We are interested in the distributions of the potential outcomes of the treated at $T=1$, that is
$F_{Y^1 \mid G, T}(y \mid 1,1)$ and $F_{Y^0 \mid G, T}(y \mid 1,1)$. For identification of these distributions, we make the following common supports assumption. 
\begin{assumption}[Common support]\label{ass:common_support} There exists a $c>0$ such that
\[
    c < \mathbb{P} (G_i = g, T_i = t) < 1-c,
\]
for $g \in \{0,1\}$ and $t \in \{0,1\}$.
\end{assumption}

Assumption \ref{ass:common_support} implies that there are both treated and control group observations in both periods. Based on this assumption,
$F_{Y^1 \mid G, T}(y \mid 1,1)$ is identified from the observed outcome for  $G=1$ at $T=1$,
$$
F_{Y^1 \mid G, T}(y \mid 1,1) = F_{Y \mid G, T}(y \mid 1,1);
$$
whereas $F_{Y^0 \mid G, T}(y \mid 1,1)$ is not identified without further assumptions. 

The  distribution of $Y^0$ conditional on $G$ and $T$ can be written as:
\begin{equation}\label{eq:dr}
F_{Y^0 \mid G, T}(y \mid g,t) = \Lambda(\alpha(y) + \beta(y) t + \gamma(y)g + \delta(y)gt), \quad y \in \mathbb{R},
\end{equation}
where $\Lambda$ is an invertible CDF on its support such as the logistic, normal or uniform, and $y \mapsto $ $(\alpha(y), $ $ \beta(y), \gamma(y), \delta(y))$ is a vector of function-valued parameters.

The representation in \eqref{eq:dr} does not make any parametric assumption about the underlying distribution of $Y^0 \mid G, T$ since the dummy variable representation within the parentheses on the right-hand side is fully saturated. The parameters of the representation are local as they vary with $y$.  To understand why \eqref{eq:dr} does not impose any restriction, note that $\alpha(y)$, $\beta(y)$, $\gamma(y)$ and $\delta(y)$ can be defined as:\footnote{See also Wooldridge (2023) equations (2.6) and (2.7).}
\[
\begin{split}
	\alpha(y) & = \Lambda^{-1} \left( F_{Y^0 \mid G, T}(y \mid 0,0) \right) \\
	\beta(y) & = \Lambda^{-1} \left( F_{Y^0 \mid G, T}(y \mid 0,1) \right) -  \Lambda^{-1} \left( F_{Y^0 \mid G, T}(y \mid 0,0) \right)   \\
	\gamma(y) & = \Lambda^{-1} \left( F_{Y^0 \mid G, T}(y \mid 1,0) \right) -  \Lambda^{-1} \left( F_{Y^0 \mid G, T}(y \mid 0,0) \right)   \\
	\delta(y) & = \Lambda^{-1} \left( F_{Y^0 \mid G, T}(y \mid 1,1) \right) -  \Lambda^{-1} \left( F_{Y^0 \mid G, T}(y \mid 1,0) \right)  \\
	& - \left[ \Lambda^{-1} \left( F_{Y^0 \mid G, T}(y \mid 0,1) \right) -  \Lambda^{-1} \left( F_{Y^0 \mid G, T}(y \mid 0,0) \right) \right].  \\
\end{split}
\]

When the link function $\Lambda$ has unbounded support, some of the previous parameters might not be well-defined for extreme values of $y$. For example, $\beta(y)$ is undetermined when $F_{Y^0 \mid G, T}(y \mid 0,0) = F_{Y^0 \mid G, T}(y \mid 0,1) =1$. The following assumption guarantees that all the parameters are well-defined on the support of $F_{Y^0 \mid G, T}(y \mid 1,1)$. 
Let $\mathcal{Y}^d_{gt}$  denote the support of $Y_d \mid G=g,T=t$, for $d, g, t \in \{0,1\}$. 
\begin{assumption}[Support Regularity]\label{ass:support} 
\[
    \mathcal{Y}^0_{11} \subseteq (\mathcal{Y}^0_{01} \cap \mathcal{Y}^0_{10} \cap \mathcal{Y}^0_{00}).
\]
\end{assumption}
Assumption \ref{ass:support} is a mild condition, which is easy to assess in practice. It holds trivially in the typical case where the support of $Y^0$ is the same across groups and over time. It does not hold, for example, when $Y^0$ is a censored variable and the smallest lower censoring point or highest upper censoring point occurs for treated group in the second period. Note that this case does not rule out applications where the treatment is an increase in the lower censoring point such as the analysis of an increase in the minimum wage on wages, because the assumption is on the support of the potential outcome without the treatment.

The following assumption is the key for identification:
\begin{assumption}[No-interaction]\label{ass:no_interaction} 
$$\delta(y)=0 \text{ for all } y \in \mathcal{Y}^0_{11} \text{ in \eqref{eq:dr}}.$$
\end{assumption}
Assumption \ref{ass:no_interaction} implies that the distribution of the potential outcome $Y^0$ should not change differently in the second period for the treatment group compared to the control group. That is, we allow a difference between the distributions of the potential outcome $Y^0$ between the treatment and control group, but this difference should be identical in both periods. This is a parallel trend type assumption on a transformation of the distribution and can be written as:
\[
\begin{split}
\Lambda^{-1}\left(F_{Y^0 \mid G, T}(y \mid 1,1)\right)  & - \Lambda^{-1}\left(F_{Y^0 \mid G, T}(y \mid 1,0)\right)  = \\
& \Lambda^{-1}\left(F_{Y^0 \mid G, T}(y \mid 0,1)\right)  - \Lambda^{-1}\left(F_{Y^0 \mid G, T}(y \mid 0,0) \right).
\end{split}
\]

This assumption depends on the link function $\Lambda$ and imposes restrictions on the distribution $F_{Y^0 \mid G, T}$ for some choices of $\Lambda$. For example, if $\Lambda$ is the identity link used in the linear probability model as in, for example, Almond et al. (2011), Dube (2019), Cengiz et al. (2019), Goodman-Bacon and Smith (2020), Goodman-Bacon (2021) and Biewen et al. (2022), one needs strong requirements in order to satisfy the parallel trends assumption (Blundell et al., 2004 and Wooldridge, 2023) That is, we need restrictions on the tails of the distribution of $F_{Y^0 \mid G, T}(y \mid 1,0)$, $F_{Y^0 \mid G, T}(y \mid 0,1)$ and $F_{Y^0 \mid G, T}(y \mid 0,0)$ to guarantee that $F_{Y^0 \mid G, T}(y \mid 1,1)$ is between $0$ and $1$. Thus, it requires that $F_{Y^0 \mid G, T}(y \mid 1,0) \leq 1 + F_{Y^0 \mid G, T}(y \mid 0,0) - F_{Y^0 \mid G, T}(y \mid 0,1)$, which might be restrictive at the top of the distribution, and $F_{Y^0 \mid G, T}(y \mid 1,0) $$\geq   F_{Y^0 \mid G, T}(y \mid 0,0) - F_{Y^0 \mid G, T}(y \mid 0,1)$, which might be restrictive at the bottom of the distribution.\footnote {For example, an increase in 0.2 in probability over time might be realistic for the control group when the initial probability was 0.5. However, if treatment group has a probability of, for example, 0.9, in the first period then it is not possible for the common trends assumption to hold.} Link functions such as the normal or logistic CDFs do not require such restrictions since the transformation expands the range of the distribution to the entire real line. 
Another requirement is that $y \mapsto \Lambda^{-1}\left(F_{Y^0 \mid G, T}(y \mid 1,0)\right) + \Lambda^{-1}\left(F_{Y^0 \mid G, T}(y \mid 0,1)\right)  - \Lambda^{-1}\left(F_{Y^0 \mid G, T}(y \mid 0,0)\right) $  be non-decreasing. These requirements could be used to develop a specification test for $\Lambda$. We provide an example of this test in the context of the empirical example in Section \ref{sec:empirical}.\footnote{Roth and Sant'Anna (2023) proposed a test for the sharp hypothesis that $y \mapsto F_{Y^0 \mid G, T}(y \mid 1,0) + F_{Y^0 \mid G, T}(y \mid 0,1) - F_{Y^0 \mid G, T}(y \mid 0,0)$ be weakly increasing, which can be used in our setting as an specification test for the identity link. We do not pursue this route as we do not encourage the use of the linear probability model.} Moreover, as in standard DiD analysis, it is possible to examine whether the ``parallel trends'' assumption holds pre-treatment when we have multiple observations in the pre-treatment period.  



Assumptions \ref{ass:common_support}--\ref{ass:no_interaction}  identify $F_{Y^0 \mid G, T}(y \mid 1,1)$ since, for $y \in \mathcal{Y}^0_{11}$,
\begin{multline}\label{eq:id}
F_{Y^0 \mid G, T}(y \mid 1,1) = \Lambda(\alpha(y) + \beta(y)  + \gamma(y) ) \\ 
=  \Lambda\left[ \Lambda^{-1}\left(F_{Y^0 \mid G, T}(y \mid 1,0)\right) + \Lambda^{-1}\left(F_{Y^0 \mid G, T}(y \mid 0,1)\right)  - \Lambda^{-1}\left(F_{Y^0 \mid G, T}(y \mid 0,0)\right) \right] \\
= \Lambda\left[ \Lambda^{-1}\left(F_{Y \mid G, T}(y \mid 1,0)\right) + \Lambda^{-1}\left(F_{Y \mid G, T}(y \mid 0,1)\right)  - \Lambda^{-1}\left(F_{Y \mid G, T}(y \mid 0,0)\right) \right],
\end{multline}
under Assumption \ref{ass:no_interaction}. The conditional distribution functions in the final equality are well defined due to the common support assumption stated in Assumption \ref{ass:common_support}. The support restrictions in Assumption \ref{ass:support} ensure that the term inside the squared brackets in \eqref{eq:id} is determined.\footnote{Note that Assumption \ref{ass:support} can be weakened to $\mathcal{Y}^0_{11} \subseteq \mathcal{Y}^0_{00}$ by working with the extended real line. In this case, for example, $\Lambda\left[ \Lambda^{-1}\left(F_{Y \mid G, T}(y \mid 1,0)\right) + \Lambda^{-1}\left(F_{Y \mid G, T}(y \mid 0,1)\right)  - \Lambda^{-1}\left(F_{Y \mid G, T}(y \mid 0,0)\right) \right] = 1$ when $F_{Y \mid G, T}(y \mid 1,0) =1$ or $F_{Y \mid G, T}(y \mid 0,1) =1$ by using the convention $\Lambda^{-1}(1) = + \infty$ and $\Lambda(+\infty) = 1$.} That is, all conditional distribution functions in the final equality are neither zero or one and hence the inverse link functions are strictly larger than minus infinity and strictly smaller than infinity. 

We present this identification result in the following lemma:
\begin{lemma}[Identification with Single Outcome]\label{lemma:did} $y \mapsto F_{Y^0 \mid G,T}(y \mid 1,1)$ is identified on $y \in \mathbb{R}$ under Assumptions \ref{ass:common_support}--\ref{ass:no_interaction}.
\end{lemma}
\begin{proof}[Proof of Lemma \ref{lemma:did}] The result follows from equation \eqref{eq:id} for $y \in \mathcal{Y}^0_{11}$. For $y \in \mathbb{R}\setminus \mathcal{Y}^0_{11}$, the result follows by extending $F_{Y^0 \mid G,T}(y \mid 1,1)$ to $0$ and $1$, that is $F_{Y^0 \mid G,T}(y \mid 1,1) = 0$ for $y \leq \inf \mathcal{Y}^0_{11}$ and $F_{Y^0 \mid G,T}(y \mid 1,1) = 1$ for $y \geq \sup \mathcal{Y}^0_{11}$.
\end{proof}
In empirical analysis, researchers typically would like to investigate objects that are related to the distributions of the potential outcome variables. One such object is the distributional treatment effect, defined as:
\[
    \tau(y) := F_{Y^1 \mid G,T}(y \mid 1,1)(y) - F_{Y^0 \mid G,T}(y \mid 1,1), \quad y \in \mathbb{R}.
\]
The distributional treatment effect measures the change in the probability that the outcome is below $y$ as a result of the treatment. Another interesting object is the quantile treatment effect, defined as:
\[
     \tau^*_q := F^{\gets}_{Y^1 \mid G,T}(q \mid 1,1)(y) - F^{\gets}_{Y^0 \mid G,T}(q \mid 1,1), \quad q \in (0,1),
\]
where $F^{\gets}(q) := \inf\{y \in \mathbb{R} : F(y) \geq q\}$ is the quantile (left-inverse) operator of $y \mapsto F(y)$ on $\mathbb{R}$. 
The quantile treatment effect measures the difference in the q-th quantile of the outcome variable as a result of the treatment. 

\subsection{Inclusion of Covariates} 

Including covariates is appealing as the assumption that $\delta(y) = 0$ may be harder to defend when there are differences in the trend between covariates related to the outcome and/or the composition of the treatment group changes over time in terms of observed characteristics; see also Melly and Santangelo (2015).   Covariates are easily incorporated into the identification result by conditioning on them and adding an overlapping support assumption. Specifically, let $X$ be a vector of covariates. The  distribution of $Y^0$ conditional on $G$, $T$ and $X$ can be written:
\begin{equation}\label{eq:dr_with_cov}
F_{Y^0 \mid G, T, X}(y \mid g,t,x) = \Lambda(\alpha(y,x) + \beta(y,x) t + \gamma(y,x)g + \delta(y,x)gt), \quad y \in \mathbb{R},
\end{equation}
where  $(y,x) \mapsto (\alpha(y,x), \beta(y,x), \gamma(y,x), \delta(y,x))$ is a vector of unspecified functions.

Let $\mathcal{Y}^d_{gtx}$ and  $\mathcal{X}_{gt}$ denote the supports of $Y_d \mid G=g,T=t, X=x$ and $X \mid  G=g,T=t$, respectively. The identifying assumptions with covariates become:
\begin{assumption}[Common support]\label{ass:common_support_with_cov} There exists a $c>0$ such that
\[
    c < \Pr (G_i = g, T_i = t \mid X_i=x) < 1-c,
\]
for $g \in \{0,1\}$, $t \in \{0,1\}$ and  $x \in \mathcal{X}_{11}$.
\end{assumption}

\begin{assumption}[Support Regularity with Covariates]\label{ass:support_with_cov} 
\[
    \mathcal{Y}^0_{11x} \subseteq (\mathcal{Y}^0_{01x} \cap \mathcal{Y}^0_{10x} \cap \mathcal{Y}^0_{00x}),  \quad x \in \mathcal{X}_{11}.
\]
\end{assumption}

\begin{assumption}[No-interaction with Covariates]\label{ass:no_interaction_with_cov} 
$$\delta(y,x)=0, \quad y \in \mathcal{Y}^0_{11x} \text{ and } x \in \mathcal{X}_{11} \text{ in \eqref{eq:dr_with_cov}.}$$
\end{assumption}

Assumptions \ref{ass:common_support_with_cov}--\ref{ass:no_interaction_with_cov} identify $F_{Y^0 \mid G, T, X}(y \mid 1,1,x)$ since, for $y \in \mathcal{Y}^0_{11}$ and $x \in \mathcal{X}_{11}$,
\begin{multline}\label{eq:id_with_cov}
F_{Y^0 \mid G, T,X}(y \mid 1,1,x) = \Lambda(\alpha(y,x) + \beta(y,x)  + \gamma(y,x) ) \\ 
=  \Lambda\left[ \Lambda^{-1}\left(F_{Y^0 \mid G, T,X}(y \mid 1,0,x)\right) + \Lambda^{-1}\left(F_{Y^0 \mid G, T,X}(y \mid 0,1,x)\right)  - \Lambda^{-1}\left(F_{Y^0 \mid G, T,X}(y \mid 0,0,x)\right) \right] \\
= \Lambda\left[ \Lambda^{-1}\left(F_{Y \mid G, T,X}(y \mid 1,0,x)\right) + \Lambda^{-1}\left(F_{Y \mid G, T,X}(y \mid 0,1,x)\right)  - \Lambda^{-1}\left(F_{Y \mid G, T,X}(y \mid 0,0,x)\right) \right],
\end{multline}
under the Assumption \ref{ass:no_interaction_with_cov}.  The conditional distribution functions after the final equality are well defined due to the common support assumption stated in Assumption \ref{ass:common_support_with_cov}. The support restrictions in Assumption \ref{ass:support_with_cov} ensure that the term between parentheses in \eqref{eq:id_with_cov} is determined. 

The following Lemma states that  $F_{Y^0 \mid G,T,X}$ is identified under the previous assumptions.
\begin{lemma}[Identification with Covariates]\label{lemma:did_with_cov} Under Assumptions  \ref{ass:support_with_cov} and \ref{ass:no_interaction_with_cov}, $(y,x) \mapsto $ \\ $F_{Y^0 \mid G,T,X}(y \mid 1,1,x)$ is identified on $(y,x) \in \mathbb{R}\times \mathcal{X}_{11}$.
\end{lemma}
\begin{proof}[Proof of Lemma \ref{lemma:did_with_cov}] The result follows from equation \eqref{eq:id_with_cov} for  $y \in \mathcal{Y}^0_{11}$. For $y \in \mathbb{R} \setminus \mathcal{Y}^0_{11}$, the result follows by extending $F_{Y^0 \mid G,T,X}(y \mid 1,1,x)$ to $0$ and $1$ as in the proof of Lemma \ref{lemma:did}.
\end{proof}
We can then identify the marginal distribution of $Y^0$ for the treated group in the second period as:
\begin{equation}\label{eq:id_with_cov2}
F_{Y^0 \mid G, T}(y \mid 1,1) = \int_{\mathcal{X}_{11}} F_{Y^0 \mid G, T,X}(y \mid 1,1,x) \mathrm{d}F_{X \mid G,T}(x \mid 1,1),    
\end{equation}
where $F_{X \mid G,T}$ is the distribution of $X$ conditional on $G$ and $T$.



\subsection{Comparison with Changes-In-Changes}
\label{ss:cic}

As our proposal provides an alternative approach to the changes-in-changes (CiC) procedure of Athey and Imbens (2006), it is useful to contrast their set up and assumptions with ours. CiC  assumes that the outcome of an individual without treatment satisfies the relationship $Y_0 = h(U,T)$ for the treatment and control groups, where $U$ is an unobserved and uniformly distributed random variable. It also assumes that $h$ is strictly increasing in the first term and that the distribution of $U$ is independent of time given the treatment outcome, i.e. $U \ci T \mid G$. Finally, the support of $U$ for the treated population should be a subset of those of the untreated population. The final assumption implies in terms of the support of the potential outcomes that: 
\[
\mY_{10}^0 \subseteq \mY_{00}^0, \qquad \mY_{11}^0 \subseteq \mY_{01}^0
\]
Their second support restriction is less restrictive than Assumption \ref{ass:support},  but we do not need their first support restriction.

Under the previous assumptions, the quantile function of $F_{Y^0 \mid G, T}(y \mid 1,1)$ is identified by:
\[
\begin{split}
F^{-1}_{Y^0 \mid G, T}(u \mid 1,1) = &\phi\left( F^{-1}_{Y^0 \mid G, T}(u \mid 1,0) \right), \\ 
& \phi(y) := F^{-1}_{Y^0 \mid G, T}\left( F_{Y^0 \mid G, T}(y \mid 0,0) \mid 0,1\right), \quad u \in [0,1],
\end{split}
\]
where it is assumed that $Y^0$ is continuous with strictly increasing distribution function. 
This identification result is based on the statistical implication of the model that the transformation $\phi$, which ``transports'' values of $Y_0$ from $T=0$ and $T=1$, is the same for the treated and and control groups.\footnote{Kim and Wooldridge (2024) show that it is possible to arrive at the same implication using weaker assumptions.} 

 In terms of the distribution of $Y^0$ in \eqref{eq:dr}, this implication imposes the restriction:
\begin{equation}\label{eq:cic-rest}
    \gamma(y) = \gamma(\phi(y)) + \delta(\phi(y)),
\end{equation}
and $\alpha(y)  = \alpha(\phi(y)) + \beta(\phi(y))$ follows from the definition of $\phi$. 
To see this, note that: 
\begin{equation}\label{eq:cic1}
    F_{Y^0 \mid G,T}(y \mid g,0) = F_{Y^0 \mid G,T}(h(h^{-1}(y,0),1) \mid g,1).
\end{equation}
Evaluating \eqref{eq:cic1} at $g=0$ and applying $F^{-1}_{Y^0 \mid G,T}(\cdot \mid 0,1)$ to both sides:
$$
h(h^{-1}(y,0),1) = F^{-1}_{Y^0 \mid G, T}\left( F_{Y^0 \mid G, T}(y \mid 0,0) \mid 0,1\right) =: \phi(y).
$$
Replacing $\phi(y)$ back in \eqref{eq:cic1} and using the representation \eqref{eq:dr}:
$$
\Lambda(\alpha(y) + \gamma(y)g) = \Lambda(\alpha(\phi(y)) +  \beta(\phi(y)) + \gamma(\phi(y))g + \delta(\phi(y))g ).
$$
The restriction then follows from equalizing the coefficients of $g$ in both sides. Note that \eqref{eq:cic-rest} is not nested with the no-interaction restriction $\delta(y)=0$ of Assumption \ref{ass:no_interaction}.

\subsection{Comparison with Roth and Sant'Anna (2023)} 

Roth and Sant'Anna (2023) derive the condition:
\begin{equation}
F_{Y^0 \mid G, T}(y \mid 1,1)  -  F_{Y^0 \mid G, T}(y \mid 1,0)  = F_{Y^0 \mid G, T}(y \mid 0,1)  - F_{Y^0 \mid G, T}(y \mid 0,0), \quad y \in \mathbb{R},
\label{eq:roth_santanna}
\end{equation}
for the parallel trends assumption in expectations:
$$
\Ep(Y^0 \mid G = 1, T = 1) - \Ep(Y^0 \mid G = 1, T = 0) = \Ep(Y^0 \mid G = 0, T = 1) - \Ep(Y^0 \mid G = 0, T = 0),
$$
to be invariant to strictly monotone transformations of $Y^0$. This condition is different from our no-interaction assumption. Indeed, our DR model with no-interaction does not generally satisfy the parallel trends assumption in expectation as:
\[
\begin{split}
\Ep(Y^0 \mid G = g, T = 1) & - \Ep(Y^0 \mid G = g, T = 0) = \\ 
& \int_{-\infty}^{\infty} [ \Lambda(\alpha(y)  + \gamma(y)g ) -  \Lambda(\alpha(y) + \beta(y) + \gamma(y)g )] \mathrm{d} y
\end{split}
\]
depends on $g$ unless $\Lambda$ is the identity map, or $\beta(y)=0$ (no trend) or $\gamma(y)=0$ (random assignment) for  $y \in \mathbb{R}$. Roth and Sant'Anna (2023) show that their condition holds if there are no trends, random assignment or a mixture of both.\footnote{See Kim and Wooldridge (2024) for a complete characterization of conditions under which parallel trends assumption in expectations hold.}


\subsection{Invariance to Strictly Monotonic Transformations}
\label{rmk:invariance} The DR model in \eqref{eq:dr} with no-interaction is invariant to strictly monotonic transformations in the sense that we specify here. If $Y^0$ follows the DR model and satisfies the no-interaction assumption, then $\tilde Y^0 = h(Y^0)$ also follows the DR model and satisfies the no-interaction assumption for any strictly monotonic transformation $h$. Common examples of these transformations include the logarithm and exponential implying that if the no-interaction assumption holds in levels also holds in logarithms and vice versa.

To see this property note that if $h$ is strictly increasing:
$$
F_{\tilde Y^0 \mid G, T, X}(\tilde y \mid g,t,x) =  \Lambda(\alpha(h^{-1}(\tilde y)) + \beta(h^{-1}(\tilde y)) t + \gamma(h^{-1}(\tilde y))g ) = \Lambda(\tilde \alpha(\tilde y) + \tilde \beta(\tilde y) t + \tilde \gamma(\tilde y)g ),
$$
where $\tilde y \mapsto h^{-1}(\tilde y)$ is the inverse function of $y \mapsto h(y)$, $\tilde \alpha = \alpha \circ h^{-1}$, $\tilde \beta = \beta \circ h^{-1}$ and $\tilde \gamma = \gamma \circ h^{-1}$.  A similar argument applies when $h$ is strictly decreasing. Unlike the parallel trends in expectation, the no-interaction or parallel trends in distribution is invariant to strictly monotonic transformations.\footnote{The distributional approach of Kim and Wooldridge (2024) also satisfies this property.}

\section{Multiple Outcomes} 

\label{sec:multiple}
Some settings may feature multiple outcomes that are potentially affected by the treatment. In these situations, we might be interested not only in how each of the outcomes is affected by the treatment, but also in how the relationship between the outcomes is affected by the treatment. For this, it is necessary to identify the joint distribution of the potential outcomes with and without treatment. We now consider a setting with two outcomes $Y$ and $Z$ and we focus on comparing features of the joint distribution of the potential outcomes with the treatment, $Y^1$ and $Z^1$, and the joint distribution of the potential outcomes without the treatment, $Y^0$ and $Z^0$, for the treated group $G=1$ in the post-treatment period $T=1$. For the sake of illustration we consider two measures of dependence. Namely, Spearman's and Kendall's rank correlation. 

Let $F_{Y^d,Z^d \mid G,T}$ be the joint distribution of $Y^d$ and $Z^d$ conditional on $G$ and $T$, and $F_{Y^d \mid G,T}$ and $F_{Z^d \mid G,T}$ be the corresponding marginals. Spearman's rank correlation between $Y^d$ and $Z^d$, $d \in \{0,1\}$,  can be expressed:
\begin{multline*}
\rho[Y^d,Z^d \mid G=1,T=1]  = \text{Corr}[F_{Y^d \mid G,T}(Y_{d} \mid 1,1), F_{Z^d \mid G,T}(Z^{d} \mid 1,1) \mid G=1,T=1] =  \\  12 \int_{-\infty}^{\infty}\int_{-\infty}^{\infty} [F_{Y^d \mid G,T}(y \mid 1,1) -1/2][F_{Z^d \mid G,T}(z \mid 1,1) - 1/2] F_{Y^d,Z^d \mid G,T}(\mathrm{d}y, \mathrm{d}z \mid 1,1); 
\end{multline*}
and Kendall's rank correlation between $Y_d$ and $Z_d$, $d \in \{0,1\}$,  can be expressed:
\begin{multline*}
\tau[Y^d,Z^d \mid G=1,T=1]  =    4 \int_{-\infty}^{\infty}\int_{-\infty}^{\infty} [F_{Y^d,Z^d \mid G,T}(y,z \mid 1,1) - 1/4] F_{Y^d,Z^d \mid G,T}(\mathrm{d}y, \mathrm{d}z \mid 1,1), 
\end{multline*}
where we assume that $Y_d$ and $Z_d$ are continuous random variables to obtain the expressions on the right hand side.

As in the univariate case, $F_{Y^1,Z^1 \mid G,T}(y,z \mid 1,1)$ is identified by the joint distribution of the observed outcomes, $F_{Y,Z \mid G,T}(y,z \mid 1,1)$, whereas $F_{Y^0,Z^0 \mid G,T}(y,z \mid 1,1)$ is not identified from the data.  To analyze identification, we use a variation of the local Gaussian representation (LGR) of a bivariate distribution from Chernozhukov, Fernand\'ez-Val and Luo (2018). Let $\Phi$ denote the Gaussian distribution function and $\Phi_2(\cdot,\cdot;\rho)$ denote the distribution of the bivariate standard normal with correlation parameter $\rho$. Moreover, $\Lambda$ is, again, a strictly increasing cumulative distribution function. As we show in Section \ref{sec:estimation}, there is a benefit of using the logistic link function in our univariate analysis. Accordingly, we employ this in our empirical analysis for estimating both the univariate and bivariate effects. 

\begin{lemma}[LGR with non-Normal Marginals]\label{lemma:lgr} The joint distribution of two random variables $Y$ and $Z$ conditional on $X$ can be represented by:
$$
F_{Y,Z \mid X}(y,z \mid x)(y,z \mid x) \equiv \Phi_2(\Phi^{-1}(\Lambda(\mu_{Y \mid X}(y \mid x))), \Phi^{-1}(\Lambda(\mu_{Z \mid X}(y \mid x))); \rho_{Y,Z \mid X}(y,z \mid x)),
$$
for all $y,z,x$, where $\mu_{Y \mid X}(y \mid x) = \Lambda^{-1}(F_{Y\mid X}(y \mid x))$, $\mu_{Z \mid X}(y \mid x) = \Lambda^{-1}(F_{Z\mid X}(z \mid x))$, and $\rho_{Y,Z \mid X}(y,z \mid x))$ is the unique solution in $\rho$ to the equation:
$$
F_{Y,Z \mid X}(y,z \mid x)(y,z \mid x) = \Phi_2(\Phi^{-1}(F_{Y\mid X}(y \mid x)), \Phi^{-1}(F_{Z\mid X}(z \mid x)); \rho).
$$
\end{lemma}
\begin{proof}
    The proof is identical to the proof of Lemma 2.1 of Chernozhukov, Fernand\'ez-Val and Luo (2018) using:
    $$
    \Phi^{-1}(\Lambda(\mu_{Y \mid X}(y \mid x))) = \Phi^{-1}(F_{Y\mid X}(y \mid x))
    $$
    and
    $$
    \Phi^{-1}(\Lambda(\mu_{Z \mid X}(z \mid x))) = \Phi^{-1}(F_{Z\mid X}(z \mid x)).    
    $$
\end{proof}

The difference between Lemma \ref{lemma:lgr} and the LGR of Chernozhukov, Fernand\'ez-Val and Luo (2018) is that the marginals are represented by a general link rather than Gaussian links, that is:
$$
F_{Y\mid X}(y \mid x)(y \mid x) \equiv \Lambda(\mu_{Y \mid X}(y \mid x)), \quad F_{Z\mid X}(z \mid x) \equiv \Lambda(\mu_{Z \mid X}(z \mid x)).
$$
We focus on the case with covariates, which covers the case without covariates as a special case by setting $X=\emptyset$. By the LGR, $F_{Y^0,Z^0 \mid G,T,X}$ can be expressed as: 
 \begin{footnotesize}{  \begin{multline}\label{eq:bdr_with_cov}
    F_{Y^0,Z^0 \mid G,T,X}(y,z \mid g,t,x) \equiv \\ \Phi_2(\Phi^{-1}(\Lambda(\mu_{Y^0 \mid G,T,X}(y \mid g,t,x))), \Phi^{-1}(\Lambda(\mu_{Z^0 \mid G,T,X}(y \mid g,t,x))); \rho_{Y^0,Z^0 \mid G,T,X}(y,z \mid g,t,x)),
\end{multline}}\end{footnotesize} 
where $\mu_{Y^0 \mid G,T,X}(y \mid g,t,x) = \alpha_Y(y,x) + \beta_Y(y,x) t + \gamma_Y(y,x) g + \delta_Y(y,x) gt$, $\mu_{Z^0 \mid G,T,X}(y \mid g,t,x) =  \alpha_Z(z,x) + \beta_Z(z,x) t + \gamma_Z(z,x) g + \delta_Z(z,x) gt$, and $\rho_{Y,Z \mid G,T,X}(y,z \mid g,t,x) = \alpha_{Y,Z}(y,z,x) + \beta_{Y,Z}(y,z,x) t + \gamma_{Y,Z}(y,z,x) g + \delta_{Y,Z}(y,z,x) gt$.

We make the following identifying assumptions with respect to the distribution function in \eqref{eq:bdr_with_cov}, which parallel the assumptions of the univariate case.
 Let $\mathcal{YZ}^d_{gtx}$ denote the support of $(Y^d,Z^d) \mid G=g,T=t, X=x$. 
\begin{assumption}[Bivariate Support Regularity]\label{ass:support_with_cov_biv} 
$$
        \mathcal{YZ}^0_{11x} \subseteq  \mathcal{YZ}^0_{01x} \cup \mathcal{YZ}^0_{10x} \cup \mathcal{YZ}^0_{00x},  \quad x \in \mathcal{X}_{11}.
$$
\end{assumption}

\begin{assumption}[Bivariate No-interaction with Covariates]\label{ass:no_interaction_with_cov_biv} 
$$\delta_Y(y,x) = \delta_Z(z,x) = \delta_{Y,Z}(y,z,x) = 0, \quad (y,z,x) \in \mathcal{YZ}^0_{11x} \times \mathcal{X}_{11} \text{ in \eqref{eq:bdr_with_cov}.}$$
\end{assumption}

\begin{lemma}[Identification with Two Outcomes and Covariates]\label{lemma:biv_with_cov} Under Assumptions \ref{ass:common_support},  \ref{ass:support_with_cov_biv} and \ref{ass:no_interaction_with_cov_biv}, $(y,z) \mapsto F_{Y^0,Z^0 \mid G,T,X}(y,z \mid 1,1,x)$ is identified on $\mathbb{R}^2 \times \mathcal{X}_{11}$.
\end{lemma}%

\begin{proof}[Proof of Lemma \ref{lemma:biv_with_cov}] We drop the covariates to lighten the notation. Under the assumptions of the Lemma, $\mu_{Y^0 \mid G,T}(y \mid g,t) = \alpha_Y(y) + \beta_Y(y) t + \gamma_Y(y) g$, $\mu_{Z^0 \mid G,T}(y \mid g,t) = \alpha_Z(z) + \beta_Z(z) t + \gamma_Z(z) g$, and $\rho_{Y,Z \mid G,T}(y,z \mid g,t) = \alpha_{Y,Z}(y,z) + \beta_{Y,Z}(y,z) t + \gamma_{Y,Z}(y,z) g$.

The parameters $\alpha_Y(y)$, $\beta_Y(y)$, $\gamma_Y(y)$, $\alpha_Z(z)$, $\beta_Z(z)$, and $\gamma_Z(z)$ are identified on $\mathcal{YZ}^0_{11}$ from the marginals of $Y$ and $Z$, by Lemma \ref{lemma:did} (or Lemma \ref{lemma:did_with_cov} when we have covariates). The parameter $\alpha_{Y,Z}(y,z)$ is identified on $\mathcal{YZ}^0_{11}$ as the solution in $\alpha$ to:
$$
F_{Y,Z \mid G,T}(y,z \mid 0,0) = \Phi_2(\Phi^{-1}(\Lambda(\alpha_Y(y))), \Phi^{-1}(\Lambda(\alpha_Z(z)))  ; \alpha).
$$
This solution exists and is unique because the RHS is strictly increasing in $\alpha$. The parameters $\beta_{Y,Z}(y,z)$ and $\gamma_{Y,Z}(y,z)$ are identified on $\mathcal{YZ}^0_{11}$ similarly as the solutions in $\beta$ and $\gamma$ of:
$$
F_{Y,Z \mid G,T}(y,z \mid 0,1) = \Phi_2(\Phi^{-1}(\Lambda(\alpha_Y(y) + \beta_Y(y))), \Phi^{-1}(\Lambda(\alpha_Z(z) + \beta_Z(z))) ; \alpha_{Y,Z}(y,z) + \beta).
$$
and
$$
F_{Y,Z \mid G,T}(y,z \mid 1,0) = \Phi_2(\Phi^{-1}(\Lambda(\alpha_Y(y) + \gamma_Y(y))), \Phi^{-1}(\Lambda(\alpha_Z(z) + \gamma_Z(z)))  ; \alpha_{Y,Z}(y,z) + \gamma).
$$
Finally, for $(y,z) \in \mathcal{YZ}^0_{11}$,
\[
\begin{split}
 F_{Y^0,Z^0 \mid G,T}(y,z \mid 1,1) & =  \Phi_2(\Phi^{-1}(\Lambda(\alpha_Y(y) + \beta_Y(y)  + \gamma_Y(y))), \Phi^{-1}(\Lambda(\alpha_Z(z) \beta_Z(z)  + \gamma_Z(z))); \\ & \alpha_{Y,Z}(y,z) +  \beta_{Y,Z}(y,z) + \gamma_{Y,Z}(y,z)).
\end{split}
\]
For $(y,z) \not\in \mathcal{YZ}^0_{11}$, we extend $F_{Y^0,Z^0 \mid G,T}(y,z \mid 1,1)$ to $0$ and $1$ similarly to the proof of Lemma \ref{lemma:did}.
\end{proof}

The marginalized distribution $F_{Y^0,Z^0 \mid G,T}(y \mid 1,1)$ is then identified by
$$
F_{Y^0,Z^0 \mid G,T}(y,z \mid 1,1) = \int_{\mathcal{X}_{11}}  F_{Y^0,Z^0 \mid G,T,X}(y,z \mid 1,1,x) \mathrm{d} F_{X \mid G,T}(x \mid 1,1).
$$

\section{Estimation}
\label{sec:estimation}
We develop estimators for the univariate and bivariate models with covariates based on distribution regression (DR). Estimators for the models without covariates can be formed using sample analogs of the expressions that identify the objects of interest.

\subsection{Univariate Case}
It is convenient to introduce a generic outcome $R$, which will be $Y$ in this subsection and either $Y$ or $Z$ in the bivariate case of next subsection.  Assume we have a sample $\{(R_{i}, X_{i}, G_{i}, T_{i}): 1\leq i \leq N\}$ of $(R, X, G, T)$. 

For estimation, we replace the  functions $(r,x) \mapsto (\alpha(r,x), \beta(r,x), \gamma(r,x))$ in \eqref{eq:dr_with_cov} by semiparametric linear indexes leading to the DR model for the conditional distribution:
\begin{equation} \label{eq:drx}
    F_{R^0 \mid G, T, X}(r \mid g,t,x) = \Lambda(p_{\alpha}(x)^\top\alpha(r) + p_{\beta}(x)^\top\beta(r) t + p_{\gamma}(x)^\top\gamma(r)g), \quad r \in \bar{\mathcal{R}}^0_{11},
\end{equation}
where $p_{\alpha}(x)$, $p_{\beta}(x)$ and $p_{\gamma}(x)$ are vectors including a constant as the first component and transformations of the covariates, $r \mapsto (\alpha(r), \beta(r), \gamma(r))$ is a vector of function-valued parameters, and $\bar{\mathcal{R}}^0_{11}$ is a strict compact subset of $\mathcal{R}^0_{11}$, the support of $R^0 \mid G=1, T=1$. 

Decompose $\alpha(r) = (\alpha_1(r),\alpha_{-1}(r)^\top)^\top$, $\beta(r) = (\beta_1(r),\beta_{-1}(r)^\top)^\top$ and $\gamma(r) = (\gamma_1(r),\gamma_{-1}(r)^\top)^\top$, where the coefficient indexed by $1$ corresponds to the intercept. To deal with the difficulties of tail estimation, we impose: for $r \in \mathcal{R}^0_{11} \setminus \bar{\mathcal{R}}^0_{11}$,
$$
\alpha_{1}(r) = \alpha_1(\bar r_r) + (r - \bar r_r) \alpha_{1\bar r_r}, \beta_1(r) = \beta_1(\bar r_r) + (r - \bar r_r) \beta_{1\bar r_r}, \gamma_1(r) = \gamma_1(\bar r_r) + (r - \bar r_r) \gamma_{1\bar r_r},
$$
and
$$
\alpha_{-1}(r) = \alpha_{-1}(\bar r_r), \quad \beta_{-1}(r) = \beta_{-1}(\bar r_r), \quad \gamma_{-1}(r) = \gamma_{-1}(\bar r_r),
$$
where $\bar{r}_r = \arg \min_{r' \in \bar{\mathcal{R}}_{11}} |r - r'|$, the point of $\bar{\mathcal{R}}^0_{11}$ that is the closest to $r$, and $\alpha_{1\bar r_r}> 0$, $\beta_{1\bar r_r} > 0$, and $\gamma_{1\bar r_r}>0$ are tail parameters. See Chernozhukov et al. (2025) for a discussion of these tail restrictions.


 We estimate all the parameters via the sequence of logit regressions at a grid of points of the support of the  outcome variable (Foresi and Peracchi, 1995, Chernozhukov, Fernandez-Val and Melly, 2013). We choose logit because it is the canonical link for binary outcomes allowing for pooled estimation of the distributions of the potential outcomes with and without the treatment (Wooldridge, 2023). In particular, we estimate the working model:
\begin{equation} \label{eq:uni}
    F_{R \mid G, T, X}(r \mid g,t,x) = \Lambda(p_{\alpha}(x)^\top\alpha(r) + p_{\beta}(x)^\top\beta(r) t + p_{\gamma}(x)^\top\gamma(r)g + p_{\theta}(x)^\top\theta(r)gt), \ \ y \in \mathcal{R}_{N,11},
\end{equation}
where $p_{\alpha}(x)$, $p_{\beta}(x)$, $p_{\gamma}(x)$, $\alpha(r)$, $\beta(r)$ and $\gamma(r)$ are the same as in the model \eqref{eq:drx},  $p_{\theta}(x)$ is a vector of transformations of the covariates to be specified below, $\theta(r)$ is a function-valued parameter that satisfies tail restrictions analogous to the ones satisfied by $\alpha(r)$, $\beta(r)$ and $\gamma(r)$, and $\mathcal{R}_{N,11}$ is a finite grid covering  $\mathcal{R}_{11}$, the support of $R \mid G=1, T=1$,  using the entire sample. Let $I_i^r := 1(R_i \leq r)$ and $\bar I_i^r = 1-I_i^r$.


\begin{algorithm}[Univariate Estimator] \label{algorithm_1}
\begin{enumerate}
\item Estimate the parameters of model \eqref{eq:uni} by unrestricted and restricted DR, that is, for $r \in \bar{\mathcal{R}}_{11}^0 \cap \mathcal{R}_{N,11}$,
\begin{footnotesize}{\begin{multline*}
    (\hat \alpha(r), \hat \beta(r), \hat \gamma(r), \hat \theta(r)) \in \arg \max_{a,b,c,d} \sum_{i=1}^N \ell_i(a,b,c,d),\\
    \ell_i(a,b,c,d) =  I_i^r \log \Lambda(p_{\alpha}(X_i)^\top a + p_{\beta}(X_i)^\top b \ T_i + p_{\gamma}(X_i)^\top c \ G_i + p_{\theta}(X_i)^\top d \ G_i T_i)  \\
     +  \bar I_i^r \log \Lambda(-p_{\alpha}(X_i)^\top a - p_{\beta}(X_i)^\top b \ T_i - p_{\gamma}(X_i)^\top c \ G_i - p_{\theta}(X_i)^\top d \ G_i T_i);
\end{multline*}}\end{footnotesize}
and for $r \in \mathcal{R}_{N,11} \setminus(\bar{\mathcal{R}}_{11}^0 \cap \mathcal{R}_{N,11})$,
\begin{multline*}
\hat \alpha_{-1}(r) = \hat \alpha_1(\bar r_r) + (r - \bar r_r) \hat \alpha_{1\bar r_r}, \quad \hat \beta_1(r) = \hat \beta_1(\bar r_r) + (r - \bar r_r) \hat \beta_{1\bar r_r}, \\
\hat \gamma_1(r) = \hat \gamma_1(\bar r_r) + (r - \bar r_r) \hat \gamma_{1\bar r_r}, \quad \hat \theta_1(r) = \hat \theta_1(\bar r_r) + (r - \bar r_r) \hat \theta_{1\bar r_r},
\end{multline*}
and
$$
\hat \alpha_{-1}(r) = \hat \alpha_{-1}(\bar r_r), \quad \hat \beta_{-1}(r) = \hat \beta_{-1}(\bar r_r), \quad \hat \gamma_{-1}(r) = \hat \gamma_{-1}(\bar r_r), \quad \hat \theta_{-1}(r) = \hat \theta_{-1}(\bar r_r),
$$
where $\bar{r}_r = \arg \min_{r' \in \bar{\mathcal{R}}_{11}} |r - r'|$,
\begin{footnotesize}{\begin{multline*}
    (\hat \alpha_{1\bar r_r}, \hat \beta_{1\bar r_r}, \hat \gamma_{1\bar r_r}, \hat \theta_{1\bar r_r}) \in \arg \max_{a,b,c,d} \sum_{i=1}^N \ell^{r_0}_i(a,b,c,d; \hat \mu^{\bar r_r}_i), \\ 
    \hat \mu^{\bar r_r}_i= p_{\alpha}(X_i)^\top \hat \alpha(\bar r_r) + p_{\beta}(X_i)^\top \hat \beta(\bar r_r) \ T_i + p_{\gamma}(X_i)^\top \hat \gamma(\bar r_r) \ G_i + p_{\theta}(X_i)^\top \hat \theta(\bar r_r) \ G_i T_i, \\
    \ell^{r_0}_i(a,b,c,d;\mu^{\bar r_r}_i) =  I_i^r \log \Lambda( (r^0 - \bar r_r) (a + b \ T_i +  c \ G_i +  d \ G_i T_i) + \mu^{\bar r_r}_i)  \\
     +  \bar I_i^r \log \Lambda(-(r^0 - \bar r_r) (a + b \ T_i +  c \ G_i +  d \ G_i T_i) - \mu^{\bar r_r}_i),
\end{multline*}}\end{footnotesize}
and $r_0 \in \mathcal{R}_{N,11} \setminus \bar{\mathcal{R}}_{11}^0$ is such that (i) there are at least $m$ observations between $\bar r_r$ and $r_0$, and greater than $r_0$ if $r_0 > \bar r_r$ (upper tail) or less than $r_0$ if $r_0 < \bar r_r$ (lower tail), and (ii) $\hat \alpha_{1\bar r_r} > 0$, $\hat \beta_{1\bar r_r} > 0$, $\hat \gamma_{1\bar r_r} > 0$ and $\hat \theta_{1\bar r_r} > 0$.
\item Construct plug-in estimators of the distributions of the potential outcomes, for $r \in \mathcal{R}_{N,11}$, 
\begin{footnotesize}{$$
\hat F_{R^0 \mid G, T}(r \mid 1,1) = \frac{1}{N_{11}} \sum_{i=1}^N G_i T_i \ \Lambda(p_{\alpha}(X_i)^\top \hat \alpha(r) + p_{\beta}(X_i)^\top \hat \beta(r) + p_{\gamma}(X_i)^\top \hat \gamma(r)),
$$}\end{footnotesize}
and
\begin{footnotesize}{$$
\hat F_{R^1 \mid G, T}(r \mid 1,1) = \frac{1}{N_{11}} \sum_{i=1}^N G_i T_i \ \Lambda(p_{\alpha}(X_i)^\top \hat \alpha(r) + p_{\beta}(X_i)^\top \hat \beta(r) + p_{\gamma}(X_i)^\top \hat \gamma(r) + p_{\theta}(X_i)^\top \hat \theta(r)),
$$}\end{footnotesize} 
where $N_{11} = \sum_{i=1}^N G_i T_i$.
\item If needed, rearrange the estimates $r \mapsto \hat F_{R^d \mid G, T}(r \mid 1,1)$ on $\mathcal{R}_{N,11}$, $d \in \{0,1\}$, to make them increasing.
\item Estimate functionals of the distributions of the potential outcome variables by using a plugin estimator, \emph{i.e.} the distributional treatment effect is estimated by
\[
    \widehat \tau(r) = \widehat F_{R^1 \mid G,T}(r \mid 1,1) - \widehat F_{R^0 \mid G,T}(y \mid 1,1), \quad r \in  \mathcal{R}_{N,11},
\]
and the quantile treatment effect is estimated by:
\[
     \widehat \tau ^*_q = \widehat F^{\gets}_{R^1 \mid G,T}(q \mid 1,1) - \widehat F^{\gets}_{R^0 \mid G,T}(q \mid 1,1), \quad q \in (0,1),
\]
where:
\[
	\widehat F^{\gets}_{R^{d} \mid G,T}(q \mid 1,1) = \min \{ r \in \mathcal{R}_{N,11}: \widehat F_{R^{d} \mid G,T}(r \mid 1,1) \leq q\} \quad d \in\{ 0,1\}.
\]
\end{enumerate}
\end{algorithm}

By the properties of the logistic link, the estimator of $F_{R^1 \mid G, T}(y \mid 1,1)$ is identical to the empirical distribution of $R$ conditional on $G=1$ and $T=1$,
$$
\hat F_{R^1 \mid G, T}(y \mid 1,1) \equiv \frac{1}{N_{11}} \sum_{i=1}^N  G_i T_i  \ 1(R_i \leq r).
$$
Note that this estimator is therefore invariant to the specification of $p_{\theta}(x)$. We set $p_{\theta}(x)=1$ to speed up computation. 
It is possible to estimate the parameters $\alpha(r)$, $\beta(r)$ and $\gamma(r)$ separately from $\theta(r)$ by using only observations with $D=0$. Our joint estimation is convenient, however, for practical inference using standard software. For example, the distributional treatment effect $\tau(y)$ corresponds to the average partial effect of $D$ in the working model \eqref{eq:uni}, so that estimates and standard errors can be obtained using existing statistical packages for logistic regression (Wooldridge, 2023).

\begin{remark}[Computation] The set $\mathcal{R}_n$ should be chosen as a fine mesh with width $\delta$ such that $\delta\sqrt{N} \to 0$. If $\mathcal{R}_n$ contains many elements, we can employ a computationally fast method similar to Chernozhukov et al. (2022) to speed-up computation. Note that the restricted optimization program to obtain $(\hat \alpha_{1\bar r_r}, \hat \beta_{1\bar r_r}, \hat \gamma_{1\bar r_r}, \hat \theta_{1\bar r_r})$ in step (1) only needs to be solved twice, once for $r_0$ in the upper tail and once for $r_0$ in the lower tail. 
Also, we recommend choosing $m \geq 30$,  which is thought to be the minimal sample size required to estimate one parameter.
\end{remark}

\subsection{Bivariate Case}

Assume we have a sample $\{(Y_{i}, Z_{i}, X_{i}, G_{i}, T_{i}): 1\leq i \leq N\}$ of $(Y, Z, X, G, T)$. For estimation, as in the univariate case, we replace the  functions in $\mu_{Y^0 \mid G,T,X}$, $\mu_{Z^0 \mid G,T,X}$ and $\rho_{Y,Z \mid G,T,X}$ by semiparametric generalized linear indexes leading to a bivariate distribution regression (BDR) model:
\begin{equation} \label{eq:bdrx1}
\mu_{Y^0 \mid G,T,X}(y \mid g,t,x) = p_{\alpha}(x)^\top \alpha_Y(y) + p_{\beta}(x)^\top \beta_Y(y) t + p_{\gamma}(x)^\top \gamma_Y(y)g,
\end{equation}
\begin{equation} \label{eq:bdrx2}
\mu_{Z^0 \mid G,T,X}(y \mid g,t,x) = q_{\alpha}(x)^\top \alpha_Z(z) + q_{\beta}(x)^\top \beta_Z(z) t + q_{\gamma}(x)^\top \gamma_Z(z)g,
\end{equation}
and
\begin{equation} \label{eq:bdrx3}
\rho_{Y^0,Z^0 \mid G,T,X}(y,z \mid g,t,x) = h(r_{\alpha}(x)^\top \alpha_{Y,Z}(y,z) + r_{\beta}(x)^\top \beta_{Y,Z}(y,z) t + r_{\gamma}(x)^\top \gamma_{Y,Z}(y,z)g),
\end{equation}
where $p_{\alpha}(x)$, $p_{\beta}(x)$, $p_{\gamma}(x)$, $q_{\alpha}(x)$, $q_{\beta}(x)$, $q_{\gamma}(x)$, $r_{\alpha}(x)$, $r_{\beta}(x)$ and $r_{\gamma}(x)$ are vectors including the covariates and their transformations, and $h(u) = \text{tanh}(u)$ is the Fisher transformation that enforces $\rho_{Y,Z \mid G,T,X}$ to lie in $[-1,1]$. 

To deal with the challenges of tail estimation, we impose the same restrictions on the parameters of $\mu_{Y^0 \mid G,T,X}$ and $\mu_{Y^0 \mid G,T,X}$ as in the univariate case. For the parameters of $\rho_{Y^0,Z^0 \mid G,T,X}$, we impose the tail restrictions, for $(y,z) \in \bar{\mathcal{Y}}_{11}^0 \times \bar{\mathcal{Z}}_{11}^0$,
$$
\alpha_{Y,Z}(y,z) = \alpha_{Y,Z}(\bar y_y, \bar z_z), \quad \beta_{Y,Z}(y,z) = \beta_{Y,Z}(\bar y_y, \bar z_z), \quad \gamma_{Y,Z}(y,z) = \gamma_{Y,Z}(\bar y_y, \bar z_z),
$$
where $\bar{\mathcal{Y}}_{11}^0$  and $\bar{\mathcal{Z}}_{11}^0$ are strict compact subsets of $\mathcal{Y}_{11}^0$  and $\mathcal{Z}_{11}^0$,  the supports of $Y^0 \mid G=1, T=1$ and $Z^0 \mid G=1, T=1$.  This restriction imposes that the local dependence paramaters remain constant outside a compact subset of the support. 

 We estimate all the parameters  using  bivariate distribution regression (Chernozhukov et al., 2025). To estimate the counterfactual distribution $F_{Y^0,Z^0 \mid G,T}(y,z \mid 1,1)$, we employ an imputation method that combines parameter estimates from the sample of the first period for both groups and the sample of the second period for the untreated group, with the sample of the covariates in the second period for the treated group. The distribution $F_{Y^1,Z^1 \mid G,T}(y,z \mid 1,1)$ is estimated using the empirical distribution of $Y$ and $Z$ in the second period for the treated group. Algorithm \ref{algorithm_2} describes the estimation procedure. Let $\mathcal{Y}^0_{N,11}$ and $\mathcal{Z}^0_{N,11}$ be finite grids covering $\mathcal{Y}^0_{11}$ and $\mathcal{Z}^0_{11}$,  $I_i^y := 1(Y_i \leq y)$,  $\bar I_i^y = 1-I_i^y$, $J_i^z := 1(Z_i \leq z)$, and $\bar J_i^z = 1-J_i^z$.


\begin{algorithm}[Bivariate Estimator] \label{algorithm_2}
\begin{enumerate}
\item 
For $y \in \mathcal{Y}^0_{N,11}$ and $z \in \mathcal{Z}^0_{N,11}$,  obtain 
$$
\hat m_i^Y(y) = \Phi^{-1}(\Lambda(p_{\alpha}(X_i)^\top \hat \alpha_Y(y) + p_{\beta}(X_i)^\top \hat \beta_Y(y) \ T_i  + p_{\gamma}(X_i)^\top \hat \gamma_Y(y) \ G_i)), 
$$
and
$$
\hat m_i^Z(z) = \Phi^{-1}(\Lambda(q_{\alpha}(X_i)^\top \hat \alpha_Z(z) + q_{\beta}(X_i)^\top \hat \beta_Z(z) \ T_i + q_{\gamma}(X_i)^\top \hat \gamma_Z(z) \ G_i)),
$$
where $\hat \alpha_Y(y)$, $\hat \beta_Y(y)$, $\hat \gamma_Y(y)$, $\hat \alpha_Z(z)$, $\hat \beta_Z(z)$ and  $\hat \gamma_Z(z)$ are the estimates of $\alpha_Y(y)$, $\beta_Y(y)$, $\gamma_Y(y)$, $\alpha_Z(z)$, $\beta_Z(z)$ and  $\gamma_Z(z)$  obtained from Algorithm \ref{algorithm_1}.
\item Estimate the parameters of the dependence function in \eqref{eq:bdrx3} by unrestricted and restricted BDR, that is, for $(y,z) \in \overline{\mathcal{YZ}}^0_{N,11} := (\mathcal{Y}^0_{N,11} \cap \bar{\mathcal{Y}}^0_{11}) \times (\mathcal{Z}^0_{N,11} \cap \bar{\mathcal{Z}}^0_{11})$,
\begin{footnotesize}{\begin{multline*}
    (\hat \alpha_{Y,Z}(y,z), \hat \beta_{Y,Z}(y,z), \hat \gamma_{Y,Z}(y,z)) \in \arg \max_{a,b,c} \sum_{i=1}^N (1 - G_i T_i) \ \ell_i(a,b,c),\\
    \ell_i(a,b,c) =  I_i^y J_i^z \log \Phi_2(\hat m_i^Y(y), \hat m_i^Z(z); h(r_{\alpha}(X_i)^\top a + r_{\beta}(X_i)^\top b \ T_i + r_{\gamma}(X_i)^\top c \ G_i))  \\
     +  I_i^y \bar J_i^z \log \Phi_2(\hat m_i^Y(y), - \hat m_i^Z(z); - h(r_{\alpha}(X_i)^\top a + r_{\beta}(X_i)^\top b \ T_i + r_{\gamma}(X_i)^\top c \ G_i))  \\
    + \bar I_i^y J_i^z \log \Phi_2(-\hat m_i^Y(y), \hat m_i^Z(z); - h(r_{\alpha}(X_i)^\top a + r_{\beta}(X_i)^\top b \ T_i + r_{\gamma}(X_i)^\top c \ G_i))  \\
    + \bar I_i^y \bar J_i^z \log \Phi_2(-\hat m_i^Y(y), -\hat m_i^Z(z); h(r_{\alpha}(X_i)^\top a + r_{\beta}(X_i)^\top b \ T_i + r_{\gamma}(X_i)^\top c \ G_i)),
\end{multline*}}
\end{footnotesize}
where $\bar y_y := \arg \min_{y' \in \bar{\mY}_n} |y-y'|$ and $\bar w_w := \arg \min_{w' \in \bar{\mW}_n} |w-w'|$; and, for $(y,z) \in (\mathcal{Y}^0_{N,11} \times \mathcal{Z}^0_{N,11}) \setminus \overline{\mathcal{YZ}}^0_{N,11}$,
$$
\hat \alpha_{Y,Z}(y,z) = \hat \alpha_{Y,Z}(\bar y_y, \bar z_z), \quad \hat \beta_{Y,Z}(y,z) = \hat \beta_{Y,Z}(\bar y_y, \bar z_z), \quad \hat \gamma_{Y,Z}(y,z) = \hat \gamma_{Y,Z}(\bar y_y, \bar z_z).
$$
\item Construct plug-in estimators of the distributions of the potential outcomes 
$$
\hat F_{Y^0, Z^0 \mid G, T}(y,z \mid 1,1) = \frac{1}{N_{11}} \sum_{i=1}^N G_i T_i \ \Phi_2(\hat m_{11i}^Y(y), \hat m_{11i}^Z(z); \hat m_{11i}^{Y,Z}(y,z)),
$$
and
$$
\hat F_{Y^1,Z^1 \mid G, T}(y,z \mid 1,1) = \frac{1}{N_{11}} \sum_{i=1}^N G_i T_i \ 1(Y_i \leq y, Z_i \leq z),
$$
where 
$$
\hat m_{11i}^Y(y) = \Phi^{-1}(\Lambda(p_{\alpha}(X_i)^\top \hat \alpha_Y(y) + p_{\beta}(X_i)^\top \hat \beta_Y(y)  + p_{\gamma}(X_i)^\top \hat \gamma_Y(y))), 
$$
$$
\hat m_{11i}^Z(z) = \Phi^{-1}(\Lambda(q_{\alpha}(X_i)^\top \hat \alpha_Z(z) + q_{\beta}(X_i)^\top \hat \beta_Z(z)  + q_{\gamma}(X_i)^\top \hat \gamma_Z(z) )),
$$
$$
\hat m_{11i}^{Y,Z}(y,z) = h(r_{\alpha}(X_i)^\top \hat \alpha_{Y,Z}(y,z) + r_{\beta}(X_i)^\top \hat \beta_{Y,Z}(y,z) + r_{\gamma}(X_i)^\top \hat \gamma_{Y,Z}(y,z))
$$
and 
$N_{11} = \sum_{i=1}^N G_i T_i$.
\item If needed, rearrange the estimate $(y,z) \mapsto \hat F_{Y^0,Z^0 \mid G, T}(y,z \mid 1,1)$ on $\mathcal{Y}_{N,11} \times \mathcal{Z}_{N,11}$,  to make it increasing in both dimensions.
\end{enumerate}
\end{algorithm}
Estimators of the functionals of the joint distributions of potential outcomes such as Spearman's and Kendall's rank correlation coefficients can be constructed using the plug-in principle.

\subsection{Bootstrap Inference}
The estimators described in Algorithms \ref{algorithm_1} and \ref{algorithm_2} can be applied to panel and repeated cross-section data. Here we describe a weighted bootstrap algorithm to perform inference on functions of the distributions of potential outcomes designed for panel data. We focus on this case because it is relevant for our empirical application below. 

To describe the procedure, we need to introduce an indicator $ID_i$, $i = 1, \ldots, N$, for the units in the panel. For example, if the sample is sorted by unit and time period, $ID = (1,1,2,2,\ldots,n,n)$, where $n=N/2$. The following algorithm describes the weighted bootstrap procedure to construct joint confidence bands for the distributions of the potential outcomes with and without the treatment in the univariate case. Inference for functionals of the distributions and for the bivariate case can be performed using similar algorithms.

 \begin{algorithm}[Weighted Bootstrap Inference for Panel Data] \label{algorithm_3}
\begin{enumerate}
\item Choose the number of bootstrap repetitions $B$, e.g., $B=500$ or $B=1,000$.
\item Draw weights for each unit independent and identically from the standard exponential distribution, independently for the data. Construct a vector of weights $ \boldsymbol{\omega} = (\omega_1,\ldots,\omega_N)$, where $\omega_i = \omega_j$ if $ID_i = ID_j$, and normalize the components of $ \boldsymbol{\omega}$ to add up to one.

\item Estimate the parameters of model \eqref{eq:uni} by unrestricted and restricted weighted DR, that is, for $r \in \bar{\mathcal{R}}_{11}^0 \cap \mathcal{R}_{N,11}$,
\begin{footnotesize}{\begin{multline*}
    (\hat \alpha^b(r), \hat \beta^b(r), \hat \gamma^b(r), \hat \theta^b(r)) \in \arg \max_{a,b,c,d} \sum_{i=1}^N \omega_i \ell_i(a,b,c,d),\\
    \ell_i(a,b,c,d) =  I_i^r \log \Lambda(p_{\alpha}(X_i)^\top a + p_{\beta}(X_i)^\top b \ T_i + p_{\gamma}(X_i)^\top c \ G_i + p_{\theta}(X_i)^\top d \ G_i T_i)  \\
     +  \bar I_i^r \log \Lambda(-p_{\alpha}(X_i)^\top a - p_{\beta}(X_i)^\top b \ T_i - p_{\gamma}(X_i)^\top c \ G_i - p_{\theta}(X_i)^\top d \ G_i T_i);
\end{multline*}}\end{footnotesize}
and for $r \in \mathcal{R}_{N,11} \setminus(\bar{\mathcal{R}}_{11}^0 \cap \mathcal{R}_{N,11})$,
\begin{multline*}
\hat \alpha^b_{-1}(r) = \hat \alpha^b_1(\bar r_r) + (r - \bar r_r) \hat \alpha^b_{1\bar r_r}, \quad \hat \beta^b_1(r) = \hat \beta^b_1(\bar r_r) + (r - \bar r_r) \hat \beta^b_{1\bar r_r},\\
\hat \gamma^b_1(r) = \hat \gamma^b_1(\bar r_r) + (r - \bar r_r) \hat \gamma^b_{1\bar r_r}, \quad \hat \theta^b_1(r) = \hat \theta^b_1(\bar r_r) + (r - \bar r_r) \hat \theta^b_{1\bar r_r},
\end{multline*}
and
$$
\hat \alpha^b_{-1}(r) = \hat \alpha^b_{-1}(\bar r_r), \quad \hat \beta^b_{-1}(r) = \hat \beta^b_{-1}(\bar r_r), \quad \hat \gamma^b_{-1}(r) = \hat \gamma^b_{-1}(\bar r_r), \quad \hat \theta^b_{-1}(r) = \hat \theta^b_{-1}(\bar r_r),
$$
where $\bar{r}_r = \arg \min_{r' \in \bar{\mathcal{R}}_{11}} |r - r'|$,
\begin{footnotesize}{\begin{multline*}
    (\hat \alpha^b_{1\bar r_r}, \hat \beta^b_{1\bar r_r}, \hat \gamma^b_{1\bar r_r}, \hat \theta^b_{1\bar r_r}) \in \arg \max_{a,b,c,d} \sum_{i=1}^N \omega_i \ell^{r_0}_i(a,b,c,d; \hat \mu^{b,\bar r_r}_i), \\ 
    \hat \mu^{b, \bar r_r}_i= p_{\alpha}(X_i)^\top \hat \alpha^b(\bar r_r) + p_{\beta}(X_i)^\top \hat \beta^b(\bar r_r) \ T_i + p_{\gamma}(X_i)^\top \hat \gamma^b(\bar r_r) \ G_i + p_{\theta}(X_i)^\top \hat \theta^b(\bar r_r) \ G_i T_i, \\
    \ell^{r_0}_i(a,b,c,d;\mu^{\bar r_r}_i) =  I_i^r \log \Lambda( (r^0 - \bar r_r) (a + b \ T_i +  c \ G_i +  d \ G_i T_i) + \mu^{\bar r_r}_i)  \\
     +  \bar I_i^r \log \Lambda(-(r^0 - \bar r_r) (a + b \ T_i +  c \ G_i +  d \ G_i T_i) - \mu^{\bar r_r}_i),
\end{multline*}}\end{footnotesize}
and $r_0$ is the same as in Algorithm \ref{algorithm_1}.
\item Construct plug-in weighted estimators of the distributions of the potential outcomes, for $r \in \mathcal{R}_{N,11}$, 
\begin{footnotesize}{$$
\hat F^b_{R^0 \mid G, T}(r \mid 1,1) = \frac{1}{N^b_{11}} \sum_{i=1}^N \omega_i G_i T_i \ \Lambda(p_{\alpha}(X_i)^\top \hat \alpha^b(r) + p_{\beta}(X_i)^\top \hat \beta^b(r) + p_{\gamma}(X_i)^\top \hat \gamma^b(r)),
$$}\end{footnotesize}
and
\begin{footnotesize}{$$
\hat F^b_{R^1 \mid G, T}(r \mid 1,1) = \frac{1}{N_{11}} \sum_{i=1}^N \omega_i G_i T_i \ \Lambda(p_{\alpha}(X_i)^\top \hat \alpha^b(r) + p_{\beta}(X_i)^\top \hat \beta^b(r) + p_{\gamma}(X_i)^\top \hat \gamma^b(r) + p_{\theta}(X_i)^\top \hat \theta^b(r)),
$$}\end{footnotesize} 
where $N^b_{11} = \sum_{i=1}^N \omega_i G_i T_i$. If needed, rearrange the estimates $r \mapsto \hat F^b_{R^d \mid G, T}(y \mid 1,1)$ on $\mathcal{R}_{N,11}$, $d \in \{0,1\}$, to make them increasing.


\item Repeat steps 1-3 $B$ times to obtain $$\left\{ \hat F_{R^0 \mid G, T}^b(r \mid 1,1), \hat F_{R^1 \mid G, T}^b(r \mid 1,1): r \in \mathcal{R}_{N,11}, 1\leq b \leq B  \right\}.$$
\item Construct an estimator of the $(1-\alpha)$-critical value of the maximal t-statistic, $\bar t_{\mathcal{R}_{11}}(1-\alpha)$, as the $(1-\alpha)$-quantile of $\{\bar t_{\mathcal{Y}_{11}}^b: 1 \leq b \leq B\}$, where
$$
\bar t_{\mathcal{R}_{11}}^b = \max_{r \in \mathcal{R}_{N,11}} \left[\frac{|\hat F_{R^0 \mid G, T}^b(r \mid 1,1) - \hat F_{R^0 \mid G, T}(r \mid 1,1)|}{S^0(r)}, \frac{|\hat F_{R^1 \mid G, T}^b(y \mid 1,1) - \hat F_{R^1 \mid G, T}(y \mid 1,1)|}{S^1(r)} \right],
$$
and $S^d(r)$ is the interquartile range of $\left\{ \hat F_{R^d \mid G, T}^b(r \mid 1,1): 1\leq b \leq B  \right\}$ divided by $1.34896$, the interquartile range of the standard normal distribution, for $d \in \{0,1\}$.
\item Construct the joint $(1-\alpha)$-confidence bands as
$$
CB_{1-\alpha}[F_{R^d \mid G, T}(\cdot \mid 1,1)] = \{\hat F_{R^d \mid G, T}(r \mid 1,1) \ \pm \ \bar t_{\mathcal{R}_{11}}(1-\alpha) S^d(r) : r \in \mathcal{R}_{N,11}\}, \quad d \in \{0,1\}.
$$
\end{enumerate}
\end{algorithm}

\begin{remark}[Empirical Bootstrap] Empirical bootstrap can be implemented by drawing the weights in step 1 from a multinomial distribution with values $1,\ldots,n$ and equal probabilities $1/n$.
\end{remark}

\begin{remark}[Repeated Cross Sectional Data] An analogous algorithm can be used for repeated cross sectional data. The only modification is that we do not impose the restriction $\omega_i = \omega_j$ if $ID_i = ID_j$ in step 2. 
\end{remark}

\section{Asymptotic theory}

\label{sec:asymptotic}
To simplify the expressions in this section we use the notation $\alpha_r$ instead of $\alpha(r)$ for the function-valued parameter vector $\alpha_r: \mathcal{R}^0_{11} \to \mathbb{R}^{d_{\alpha}}$, where $d_{\alpha}$ is the dimension of the vector $p_{\alpha}(x)$.
We adopt similar notation for all the function-valued parameters. It is also convenient to introduce the following notation for the tails. Let $\bar r_r  := \arg \min_{r' \in \bar{\mR}_{11}^0} |r-r'|$ for $r \in \{y,w\}$ and $\bar \mR_{11}^0 \in \{\bar \mY_{11}^0,\bar \mW_{11}^0\}$. Note that $\bar r_r = r$ if $r \in \bar \mR_{11}^0$, $\bar r_r = \bar r$ if $r \ge \bar r$ and $\bar r_r = \underline r$ if $r \leq \underline r$, where $\bar r := \sup (\bar \mR_{11}^0)$ and $\underline r := \inf(\bar \mR_{11}^0)$. 
We also use the following notation for partial derivatives: $\partial_x f(x) := \partial f(x) / \partial x$ and $\partial_{xx} f(x) := \partial^2 f(x) / (\partial x \partial x^\top )$. We provide asymptotic theory for the panel data case. Theory for repeated cross-sectional data can be obtained directly from Lemmas 6.1 and B.4 in Chernozhukov et al. (2013) in the univariate case and Corollary 5.1 in Chernozhukov et al. (2025) in the bivariate case.

\subsection{Univariate case}

We modify the notation to account for panel-data structure.
Following the panel data convention we denote by $V_t$ the value of the random variable $V$ at time $t \in \{0,1\}$. We observe a sample $\{(R_{i0}, R_{i1} X_{i0}, X_{i1}, G_{i}), 1 \leq i \leq n\}$ composed of $i.i.d.$ copies of $(R_0, R_1, X_0, X_1, G)$. This allows for arbitrary dependence between the observations of each unit over time. 

Using this notation, let $I^r_t := \textbf{1}(R_t \leq r)$ for $t \in \{0,1\}$. 
For $r \in \bar{\mR}_{11}^0$, the  probability mass function of $I^r_t$  conditional on $G$ and $X_t$ is
\begin{multline*}
	f_{I_t^r \mid G, X_t}(i_t \mid g, t, x_t)  = \Lambda\left( p_{\alpha}(x_t)^\top\alpha_r  + p_{\beta}(x_t)^\top\beta_r g + p_{\gamma}(x_t)^\top\gamma_r t + p_{\theta}(x_t)^\top\theta_r g t \right)^{i_t}  \\ \times \Lambda\left( -p_{\alpha}(x_t)^\top\alpha_r  - p_{\beta}(x_t)^\top\beta_r g - p_{\gamma}(x_t)^\top\gamma_r t - p_{\theta}(x_t)^\top\theta_r g t  \right)^{1-i_t}.     
\end{multline*}
Let the average partial log-likelihood be
\[
	\ell^r(\zeta_r) =  \frac{1}{n} \sum_{i=1}^n 	\ell_i^r(\zeta_r), \quad \ell_i^r(\zeta_r) := \sum_{t=0}^1 \log f_{I_t^r \mid G, X_t}(i_t \mid g, t, x_t),
\]
where $\zeta_r = (\alpha_r, \beta_r, \gamma_r, \theta_r)$.

Similarly at the tails, for $r \notin \bar{\mR}_{11}^0$, 
\begin{equation}\label{eq:plik-univ}
    	\ell^{r}(\xi_{\bar r_r}, \zeta_{\bar r_r}) =  \frac{1}{n} \sum_{i=1}^n 	\ell_i^{r}(\xi_{\bar r_r}, \zeta_{\bar r_r}), \quad \ell_i^{r}(\xi_{\bar r_r}, \zeta_{\bar r_r}) := \sum_{t=0}^1 \log f_{I_t^{r} \mid G, X_t}(i_t \mid g, t, x_t),
\end{equation}
where $\xi_{\bar r_r} := (\alpha_{1 \bar r_r}, \beta_{1 \bar r_r}, \gamma_{1 \bar r_r}, \theta_{1 \bar r_r})$,
\begin{multline*}
	f_{I_t^{r} \mid G, X_t}(i_t \mid g, t, x_t)  = \Lambda\left( (r-\bar r_r)(\alpha_{1\bar r_r} + \beta_{1\bar r_r}g + \gamma_{1\bar r_r}t + \theta_{1\bar r_r} gt) + \mu_{\bar r_r}(g,t,x_t) \right)^{i_t}  \\ \times \Lambda\left(-(r-\bar r_r)(\alpha_{1\bar r_r} + \beta_{1\bar r_r}g + \gamma_{1\bar r_r}t + \theta_{1\bar r_r} gt) - \mu_{\bar r_r}(g,t,x_t)  \right)^{1-i_t},     
\end{multline*}
and $\mu_{\bar r_r}(g,t,x_t) := p_{\alpha}(x_t)^\top\alpha_{\tilde r}  + p_{\beta}(x_t)^\top\beta_{\tilde r} g + p_{\gamma}(x_t)^\top\gamma_{\tilde r} t + p_{\theta}(x_t)^\top\theta_{\tilde r} g t$.

\begin{assumption}[DR Model]\label{ass:dr} (1) The conditional distribution function $F_{R_t \mid G, X_t}(y \mid g, x) = F_{R \mid G, T, X}(r \mid g,t,x)$ takes the form as in \eqref{eq:uni} with  $\Lambda(u) = (1+\exp(-u))^{-1}$, the  standard logistic distribution, and satisfies the tail restrictions. (2) The set $\mathcal{X}_{11}$ is  compact.  (3)  The set  $\mathcal{R}_{11}^0$ is either an  open interval in $\mathbb{R}$ or a finite set in $\mathbb{R}$. In the former case, the conditional density function $f_{R_t \mid G, X_t}(r \mid g, x)$ exists, is uniformly bounded and uniformly continuous on $(r,x)$ in the support of $(Y_t,X_t) \mid G=g$ for $t\in\{0,1\}$ and $ g \in \{0,1\}$. (4) The minimum eigenvalues of $\Ep  \partial_{\zeta \zeta}  \ell_i^r(\zeta^0_r)$ and  $\Ep
     \partial_{\xi \xi}\ell_i^{r_0}(\xi^0_{\bar r_{r_0}}, \zeta^0_{\bar r_{r_0}})$ are bounded away from zero uniformly over $r \in \bar{\mR}_{11}^0$ and some $r_0 \notin \bar{\mR}_{11}^0$, where $\zeta^0_r$ and $\xi^0_{\bar r_r}$ denote the true value of the parameters $\zeta_r$ and $\xi_{\bar r_r}$. 
\end{assumption}

\begin{assumption}[Exchangeable Bootstrap]\label{ass:eb} For each $n$, $(\omega_{n_1}, ..., \omega_{nn})$ is an exchangeable,\footnote{A sequence of random variables $X_1, X_2, ..., X_n$ is exchangeable if for any finite permutation $\sigma$ of the indices $1,2, ..., n$ the joint
distribution of the permuted sequence $X_{\sigma(1)}, X_{\sigma(2)},
...,X_{\sigma(n)} $ is the same as the joint distribution of the original sequence.} nonnegative random vector, which is independent of the data, such that for some $\epsilon> 0$ 
\begin{equation}  \label{eq: assumptions weighted bootstrap}
\begin{split}
		\sup_{n} \Ep[(\omega_{n})^{2+\epsilon}] < \infty, \ \ n^{-1}\sum
		_{i=1}^{n} \left( \omega^d_{ni} - \bar{\omega}_{n} \right)^{2} \to_{\Pr} 1, \ \
		\bar \omega_{n} \to_{\Pr} 1,
\end{split}
\end{equation}
where $\bar \omega_{n} = n^{-1} \sum_{i=1}^{n} \omega^d_{n_i} $. 
\end{assumption}

Let $\hat \zeta_r$ and $\hat \zeta^*_r $ be the estimator and the corresponding bootstrap draw of the parameter $\zeta_r$ for $r \in \mathcal{R}_{11}^0$, and $\hat \xi_{\bar r_r}$ and $\hat \xi^*_{r_r}$ be the estimator and the corresponding bootstrap draw of the parameter $\xi_{\bar r_r}$ for $r \notin \mathcal{R}_{11}^0$. We follow the notation of van der Vaart and Wellner (2013) for weak convergence and bootstrap consistency. Thus,  $\mathbb{Z}_n \rightsquigarrow \mathbb{Z}$ in $\mathbb{E}$ denotes weak convergence of a stochastic process $\mathbb{Z}_n$ to a random element $\mathbb{Z}$ in a normed space $\mathbb{E}$.
Let $D_{n}$ denote the data
vector and $M_{n}$ be the vector of random variables used to generate
bootstrap draws given $D_{n}$. Consider the random element $%
\mathbb{Z}^{*}_{n} = \mathbb{Z}_{n}(D_{n}, M_{n})$. We say that the bootstrap law of $\mathbb{Z}^{*}_{n}$
consistently estimates the law of some tight random element $\mathbb{Z}$ and
write $\mathbb{Z}^{*}_{n} \rightsquigarrow_{\mathrm{P}} \mathbb{Z} $ in $\mathbb{E}$
if: 
\begin{equation*}  \label{boot1}
\begin{array}{r}
\sup_{h \in\text{BL}_{1}(\mathbb{E})} \left| \Ep_{M_{n}} h \left( \mathbb{Z}%
^{*}_{n}\right) - \Ep h(\mathbb{Z})\right| \rightarrow_{\mathrm{P}} 0,%
\end{array}%
\end{equation*}
where $\text{BL}_{1}(\mathbb{E})$ denotes the space of functions with
Lipschitz norm at most 1 and $\Ep_{M_{n}}$ denotes the conditional expectation
with respect to $M_{n}$ given the data $D_{n}$; and $ \rightarrow_{\mathrm{P}}$ denotes
convergence in (outer) probability.


\begin{lemma}[FCLT and Bootstrap CLT for $\hat \zeta_r$ and $\hat \xi_{\bar r_r}$] \label{lemma:asymptotic_univariate} Let $\{W_i:=(R_{0i}, R_{1i},G_i, X_{0i}, X_{1i}) : 1 \leq i \leq n\}$ be a sample of i.i.d. copies of the random vector $W:=(R_0, R_1, G, X_0, X_1)$ that has probability law $\mathrm{P}$ and obeys Assumption \ref{ass:dr}. (i) As $n \rightarrow \infty$ the DR coefficient process possesses the following first order approximation and limit law:
\[
	\sqrt{n} \left( \widehat{\zeta}_r - \zeta_r \right) 
    \rightsquigarrow
	\mathbb{Z}_r^{\zeta} :=  -\Ep \left( \partial_{\zeta \zeta}  \ell_i^r(\zeta_r)\right)^{-1} \mathbb{G}\partial_{\zeta} \ell_i^r(\zeta_r)  \text{ in } \ell^{\infty}(\bar{\mathcal{R}}_{11}^0)^{d_{\zeta}},
\]
and, for  $r_0 \notin \bar{\mR}_{11}^0$,  in $\mathbb{R}^{d_{\xi}}$,
$$
\sqrt{n}(\hat \xi_{\bar r_{r_0}} -  \xi_{\bar r_{r_0}}) \leadsto \mathbb{Z}^{\xi}_{\bar r_{r_0}} := \Ep\left( 
     \partial_{\xi \xi}\ell_i^{r_0}(\xi_{\bar r_{r_0}}, \zeta_{\bar r_{r_0}}) \right)^{-1}\left[\mathbb{G}\partial_{\xi}\ell_i^{r_0}(\xi_{\bar r_{r_0}}, \zeta_{\bar r_{r_0}}) + \Ep\left( \partial_{\xi \zeta}\ell_i^{r_0}(\xi_{\bar r_{r_0}}, \zeta_{\bar r_{r_0}})\right) \mathbb{Z}^{\zeta}_{\bar r_{r_0}} \right],$$
where  $d_{\zeta} := \dim \zeta_r$, $d_{\xi} := \dim \xi_{\bar r_r}$,  $\mathbb{G}$ is a $\mathrm{P}$-Brownian bridge.\footnote{A zero-mean Gaussian process $\mathbb{G}$ is a $\mathrm{P}$-Brownian bridge if its covariance function takes the form $\mathbb{E}\left(\mathbb{G}(f)\mathbb{G}(g)\right) = \int fg d\mathrm{P} - \int fd\mathrm{P} \int gd\mathrm{P}$. See Van der Vaart (1998), page 269.} 
(ii) If, in addition, Assumption \ref{ass:eb} holds, the exchangeable bootstrap law is consistent for the limit law, namely, as $n \rightarrow \infty$, $\sqrt{n} \left( \widehat{\zeta}^*_r - \hat \zeta_r \right) \rightsquigarrow_{\mathrm{P}} \mathbb{Z}_r^{\zeta}$  in  $\ell^{\infty}(\bar{\mathcal{R}}_{11}^0)^{d_{\zeta}}$ and $\sqrt{n}(\hat \xi^*_{\bar r_{r_0}} - \hat \xi_{\bar r_{r_0}}) \leadsto_{\mathrm{P}} \mathbb{Z}^{\xi}_{\bar r_{r_0}}$ in $\mathbb{R}^{d_{\xi}}$, for $r_0 \notin \bar{\mR}_{11}^0$.
\end{lemma}
\begin{proof} The proof is in Appendix \ref{app:asymptotic_univariate}.
\end{proof}

The distributional and the quantile treatment effects are functionals of the parameter $(\zeta_r,\xi_{\bar r_{r_0}})$ and the distribution of $X$. For example, for $d \in\{0,1\}$, $r \in \bar{\mR}_{11}^0$ and $p_{\zeta}(d,x) := (p_{\alpha}(x)^\top, p_{\beta}(x)^\top, p_{\gamma}(x)^\top, p_{\theta}(x)^\top d )^\top$,
\[
    F_{R^d \mid T,G}(r \mid 1,1) = \int_{\mathcal{X}_{11}} \Lambda\left(  p_{\zeta}(d,x)^\top \zeta_r  \right) \mathrm{d} F_{X \mid T,G}(x \mid 1,1),
\]
whereas for $r \not\in \bar{\mR}_{11}^0$ and $p_{\xi}(d) = (1, 1, 1, d )^\top$,
\[
    F_{R^d \mid T,G}(r \mid 1,1) = \int_{\mathcal{X}_{11}} \Lambda\left( (r-\bar r_r) p_{\xi}(d)^\top \xi_{\bar r_{r_0}} + p_{\zeta}(d,x)^\top \zeta_{\bar r_r}  \right) \mathrm{d} F_{X \mid T,G}(x \mid 1,1),
\]
which are both counterfactual distributions. 
The corresponding estimators are constructed using the plug-in rule, that is
\[
    \hat F_{R^d \mid T,G}(r \mid 1,1) = \int_{\mathcal{X}_{11}} \Lambda\left( p_{\zeta}(d,x)^\top  \hat \zeta_r\right) \mathrm{d} \hat F_{X \mid T,G}(x \mid 1,1), \quad r \in \bar{\mR}_{11}^0,
\]
and
\[
    \hat F_{R^d \mid T,G}(r \mid 1,1) = \int_{\mathcal{X}_{11}} \Lambda\left( (r-\bar r_r) p_{\xi}(d)^\top \hat \xi_{\bar r_{r_0}} + p_{\zeta}(d,x)^\top \hat \zeta_{\bar r_r}\right) \mathrm{d} \hat F_{X \mid T,G}(x \mid 1,1), \quad r \notin \bar{\mR}_{11}^0,
\]
where $\hat F_{X \mid T,G}(\cdot \mid 1,1)$ is the empirical distribution of $X$ conditional on $G=1$ and $T=1$. 

The following result follows from the theory of estimated counterfactual distributions in Chernozhukov et al. (2013).

\begin{theorem}[FCLT and Bootstrap CLT for $\hat F_{R^d \mid T,G}$]\label{thm:fclt_cdist} Assume the conditions of Lemma \ref{lemma:asymptotic_univariate} hold. (i) As $n \to \infty$, 
    $$
    \left\{\sqrt{n}\left( \hat F_{R^d \mid T,G}(r \mid 1,1) - F_{R^d \mid T,G}(r \mid 1,1) \right): d \in\{0,1\} \right\} \rightsquigarrow \left\{\mathbb{Z}_r^{F_d} :  d \in\{0,1\}\right\} \text{ in } \ell^{\infty}(\mathcal{R}_{11}^0)^2,
    $$
    where
    \begin{multline*}
        \mathbb{Z}_r^{F_d} = \sqrt{\pi_1} \mathbb{G} \left[ \Lambda\left( p_{\zeta}(d,X_1)^\top  \zeta_r \right) G\right] \\ +\int_{\mathcal{X}_{11}} \lambda\left( p_{\zeta}(d,x)^\top \zeta_r \right)\left[ p_{\zeta}(x,d)^\top \mathbb{Z}_{\bar r_r}^{\zeta} +  (r - \bar r_r) p_{\xi}(d)^\top \mathbb{Z}^{\xi}_{\bar r_{r_0}} \right] \mathrm{d} F_{X \mid T,G}(x \mid 1,1),
    \end{multline*}
   $\pi_1 = \mathrm{P}(G = 1)$, and $\zeta_r = \zeta_{\bar r_r} + (r-\bar r_r)[(\alpha_{1\bar r_r},0\alpha_{-1\bar r_r}), (\beta_{1\bar r_r},0\beta_{-1\bar r_r}), (\gamma_{1\bar r_r},0\gamma_{-1\bar r_r}), (\theta_{1\bar r_r},0\theta_{-1\bar r_r})] $. 
    (ii) As $n \to \infty$, 
    $$
    \sqrt{n}\left( \hat F^*_{R^d \mid T,G}(r \mid 1,1) - \hat F_{R^d \mid T,G}(r \mid 1,1) : d \in\{0,1\} \right) \rightsquigarrow_{\mathrm{P}} (\mathbb{Z}_r^{F_d} :  d \in\{0,1\}) \text{ in } \ell^{\infty}(\mathcal{R}_{11}^0)^2.
    $$
\end{theorem}

Similar results  for distributional and quantile treatment effects can be obtained from Theorem \ref{thm:fclt_cdist} together with the functional delta method; see Chernozhukov et al. (2013).

\subsection{Bivariate case}


We observe a sample $\{W_i := (Y_{i0}, Y_{i1}, Z_{i0}, Z_{i1}, X_{i0}, X_{i1}, G_i), 1 \leq i \leq n\}$ composed of \textit{i.i.d.} copies of the random vector $W := (Y_{0}, Y_{1}, Z_{0}, Z_{1}, X_{0}, X_{1}, G)$. Define $I^y_t := \textbf{1}(Y_t \leq y)$ and $J^z_t = \textbf{1}(Z_t \leq z)$ for $t=\{0,1\}$. For $y \in \bar{\mY}_{11}^0$ and $z \in \bar{\mZ}_{11}^0$, the average partial log-likelihood is
\begin{equation}\label{eq:plik-biv}
    \ell^{yz}(\zeta_{yz}) =  \frac{1}{n} \sum_{i=1}^n 	\ell_i^{yz}(\zeta_{yz}), \quad  \ell_i^{yz}(\zeta_{yz}) := \sum_{t=0}^1  \log f_{I_t^y, J_t^z \mid G, X_t}(I_{it}^y, J_{it}^z \mid G_i, X_{it}),
\end{equation}
where $\zeta_{yz}  := \left(\alpha_{y}, \beta_{y}, \gamma_{y}, \alpha_{z}, \beta_{z}, \gamma_{z}, \alpha_{yz}, \beta_{yz}, \gamma_{yz}\right),$
\begin{equation*}
	\begin{split}
		f_{I_t^y, J_t^z \mid G, X_t}&(i, j \mid g, x) =
		\Phi_{2}(\Xi(\mu_y(g, t, x)), \Xi(\mu_z(g, t, x)); \rho_{yz}(g, t,x))^{i j} \\
		&\times \Phi_{2}(\Xi(\mu_y(g, t, x)), -\Xi(\mu_z(g, t, x)); -\rho_{yz}(g, t,x))^{i(1-j)}  \\
		&\times \Phi_{2}(-\Xi(\mu_y(g, t, x)), \Xi(\mu_z(g, t, x)); -\rho_{yz}(g, t,x))^{(1-i)j} \\
		&\times \Phi_{2}(-\Xi(\mu_y(g, t, x)), -\Xi(\mu_z(g, t, x)); \rho_{yz}(g, t,x))^{(1-i)(1-j)}, 
	\end{split}
\end{equation*}
$\Xi := \Phi^{-1} \circ \Lambda$, $\mu_y(g, t, x) := p_{\alpha}(x_t)^\top\alpha_{y}  + p_{\beta}(x_t)^\top\beta_{y} g + p_{\gamma}(x_t)^\top\gamma_{y} t$, $\mu_z(g, t, x) := q_{\alpha}(x_t)^\top\alpha_{z}  + q_{\beta}(x_t)^\top\beta_{z} g + q_{\gamma}(x_t)^\top\gamma_{z} t$, and $\rho_{yz}(g, t,x):= h(r_{\alpha}(x_t)^\top\alpha_{yz}  + r_{\beta}(x_t)^\top\beta_{yz} g + r_{\gamma}(x_t)^\top\gamma_{yz} t)$.

\begin{assumption}[BDR Model]\label{ass:bdr} (1) The joint distribution of $(Y_t, Z_t)$ conditional on $G, X_t$, $t \in \{0,1\}$ follows the BDR model \eqref{eq:bdr_with_cov} with the tail restrictions. (2) Assumption \ref{ass:dr} holds for $R \in \{Y,Z\}$.  (3) The equation $\Ep[\partial_{\zeta_{yz}} \ell_i^{yz}(\zeta)] = 0$ possesses a unique solution at $\zeta^{0}_{yw}$    and the minimum eigenvalue of the matrix $\Ep[\partial_{\zeta \zeta } \ell_i^{yz}(\zeta_{yz})]$ is bounded away from zero uniformly over $(y,z) \in \overline{\mY\mZ}_{11}^0 := \bar{\mY}_{11}^0 \times\bar{\mZ}_{11}^0$.
\end{assumption}
%
Let $\hat \zeta_{yz}$ and $\hat \zeta^*_{yz} $ be the estimator and the corresponding bootstrap draw of the parameter $\zeta_{yz}$ for $(y,z) \in \overline{\mY\mZ}_{11}^0$, , and $\hat \xi_{\bar r_r}$ and $\hat \xi^*_{r_r}$ be the estimator and the corresponding bootstrap draw of the parameter $\xi_{\bar r_r}$ for $r \notin \mathcal{R}_{11}^0$, with $r \in \{y,z\}$ and $\mathcal{R} \in \{\mY,\mZ\}$.

\begin{lemma}[FCLT and Bootstrap CLT for $\hat \zeta_{yz}$, $\hat \xi_{\bar y_y}$ and $\hat \xi_{\bar z_z}$]\label{lemma:asymptotic_bivariate} Let $\{W_i:=(Y_{0i}, Y_{1i},Z_{0i}, Z_{1i},G_i, X_{0i}, X_{1i}) : 1 \leq i \leq n\}$ be a sample of i.i.d. copies of the random vector $W:=(Y_0, Y_1,Z_0, Z_1, G, X_0, X_1)$ that has probability law $\mathrm{P}$ and obeys Assumption \ref{ass:bdr}. (i) As $n \rightarrow \infty$ the DR coefficient process possesses the following first order approximation and limit law:
\[
	\sqrt{n} \left( \widehat{\zeta}_{yz} - \zeta_{yz} \right) 
    \rightsquigarrow
	\mathbb{Z}_{yz}^{\zeta} :=  -\Ep \left( \partial_{\zeta \zeta}  \ell_i^{yz}(\zeta_{yz})\right)^{-1} \mathbb{G}\partial_{\zeta} \ell_i^{yz}(\zeta_{yz})  \text{ in } \ell^{\infty}(\overline{\mY\mZ}_{11}^0)^{d_{\zeta}},
\]
and, for  $r_0 \notin \bar{\mR}_{11}^0$ with $r \in \{y,z\}$ and $\mR \in \{\mY,\mZ\}$, in  $\mathbb{R}^{d_{\xi}}$,
$$
\sqrt{n}(\hat \xi_{\bar r_{r_0}} -  \xi_{\bar r_{r_0}}) \leadsto \mathbb{Z}^{\xi}_{\bar r_{r_0}} := \Ep\left( 
     \partial_{\xi \xi}\ell_i^{r_0}(\xi_{\bar r_{r_0}}, \zeta_{\bar r_{r_0}}) \right)^{-1}\left[\mathbb{G}\partial_{\xi}\ell_i^{r_0}(\xi_{\bar r_{r_0}}, \zeta_{\bar r_{r_0}}) + \Ep\left( \partial_{\xi \zeta}\ell_i^{r_0}(\xi_{\bar r_{r_0}}, \zeta_{\bar r_{r_0}})\right) \mathbb{Z}^{\zeta}_{\bar r_{r_0}} \right],$$
where  $d_{\zeta} := \dim \zeta_{yz}$, $d_{\xi} := \dim \xi_{\bar r_{r_0}}$,  $\mathbb{G}$ is a $\mathrm{P}$-Brownian bridge.\footnote{A zero-mean Gaussian process $\mathbb{G}$ is a $\mathrm{P}$-Brownian bridge if its covariance function takes the form $\mathbb{E}\left(\mathbb{G}(f)\mathbb{G}(g)\right) = \int fg d\mathrm{P} - \int fd\mathrm{P} \int gd\mathrm{P}$. See Van der Vaart (1998), page 269.} 
(ii) If, in addition, Assumption \ref{ass:eb} holds, the exchangeable bootstrap law is consistent for the limit law, namely, as $n \rightarrow \infty$, $\sqrt{n} \left( \widehat{\zeta}^*_{yz} - \hat \zeta_{yz} \right) \rightsquigarrow_{\mathrm{P}} \mathbb{Z}_{yz}^{\zeta}$  in  $\ell^{\infty}(\overline{\mY\mZ}_{11}^0)^{d_{\zeta}}$ and $\sqrt{n}(\hat \xi^*_{\bar r_r} - \hat \xi_{\bar r_r}) \leadsto_{\mathrm{P}} \mathbb{Z}^{\xi}_{\bar r_r}$ in $\mathbb{R}^{d_{\xi}}$, for $r \notin \bar{\mR}_{11}^0$.
\end{lemma}
\begin{proof} The proof is in Appendix \ref{appendix:asymptotic_bivariate}.
\end{proof}

Kendall coefficient, Spearman coefficient and other quantities of interest are functions of the joint distribution of the potential outcomes $Y^0$ and $Z^0$, which is a functional of the of the parameter $(\zeta_{yz},\xi_{\bar y_{y_0}},\xi_{\bar z_{z_0}})$ and the distribution of $X$. For example, for  $(y,z) \in \overline{\mY\mZ}_{11}^0$, $P(x) := (p(x)^\top, 0,\ldots,0)^\top$, $p(x) := (p_{\alpha}(x)^\top, p_{\beta}(x)^\top, p_{\gamma}(x)^\top)^\top$, $Q(x) := (0,\ldots, q(x)^\top,0,\ldots,0)^\top$, $q(x) := (q_{\alpha}(x)^\top, q_{\beta}(x)^\top, q_{\gamma}(x)^\top)^\top$, $R(x) := (0,\ldots 0,r(x)^\top)^\top$ and $r(x) := (r_{\alpha}(x)^\top, r_{\beta}(x)^\top, r_{\gamma}(x)^\top)^\top$,
\[
    F_{Y^0,Z^0 \mid T,G}(y,z \mid 1,1) = \int_{\mathcal{X}_{11}} \Phi_{2}(\Xi(P(x)^\top \zeta_{yz}), \Xi(Q(x)^\top \zeta_{yz}); \rho_{yz}(R(x)^\top \zeta_{yz})) \mathrm{d} F_{X \mid T,G}(x \mid 1,1),
\]
whereas for  $(y,z) \not\in \overline{\mY\mZ}_{11}^0$,
$$
        F_{Y^0,Z^0 \mid T,G}(y,z \mid 1,1) = \int_{\mathcal{X}_{11}} \Phi_{2}( \Xi(m_{11y}^Y), \Xi(m_{11z}^Z ); \rho_{yz}(m_{11yz}^{YZ})) \mathrm{d} F_{X \mid T,G}(x \mid 1,1),
$$
where $m_{11y}^Y := (y-\bar y_y) p_{\xi}(0)^\top \xi_{\bar y_{y_0}} + P(x)^\top \zeta_{\bar y_y \bar z_z}$, $m_{11z}^Z := (z-\bar z_z) p_{\xi}(0)^\top \xi_{\bar z_{z_0}} + Q(x)^\top \zeta_{\bar y \bar z}$ and $m_{11yz}^{YZ} :=R(x)^\top \zeta_{\bar y_y \bar z_z}$. Both  are counterfactual distributions. 
The corresponding estimators are constructed using the plug-in rule, that is
\[
    \hat F_{Y^0,Z^0 \mid T,G}(y,z \mid 1,1) = \int_{\mathcal{X}_{11}} \Phi_{2}(\Xi(P(x)^\top \hat \zeta_{yz}), \Xi(Q(x)^\top \hat \zeta_{yz}); \rho_{yz}(R(x)^\top \hat \zeta_{yz})) \mathrm{d} \hat F_{X \mid T,G}(x \mid 1,1), 
\]
and
$$
 F_{Y^0,Z^0 \mid T,G}(y,z \mid 1,1) = \int_{\mathcal{X}_{11}} \Phi_{2}( \Xi(\hat m_{11y}^Y), \Xi(\hat m_{11z}^Z ); \rho_{yz}(\hat m_{11yz}^{YZ})) \mathrm{d} F_{X \mid T,G}(x \mid 1,1)
$$
where $\hat F_{X \mid T,G}(\cdot \mid 1,1)$ is the empirical distribution of $X$ conditional on $G=1$ and $T=1$, $\hat m_{11y}^Y := (y-\bar y_y) p_{\xi}(0)^\top \hat \xi_{\bar y_{y_0}} + P(x)^\top \hat \zeta_{\bar y_y \bar z_z}$, $\hat m_{11z}^Z := (z-\bar z_z) p_{\xi}(0)^\top \hat \xi_{\bar z_{z_0}} + Q(x)^\top \hat \zeta_{\bar y \bar z}$ and $m_{11yz}^{YZ} :=R(x)^\top \hat \zeta_{\bar y_y \bar z_z}$. 

The following result follows from the theory of estimated counterfactual distributions in Chernozhukov et al. (2013).

\begin{theorem}[FCLT and Bootstrap CLT for $\hat F_{Y^0,Z^0 \mid T,G}$]\label{thm:fclt_cdist_biv} Assume the conditions of Lemma \ref{lemma:asymptotic_bivariate} hold. (i) As $n \to \infty$, 
    $$
    \sqrt{n}\left( \hat F_{Y^0,Z^0 \mid T,G}(y,z \mid 1,1) - F_{Y^0,Z^0 \mid T,G}(r \mid 1,1)  \right) \rightsquigarrow \mathbb{Z}_{yz}^{F_0} \text{ in } \ell^{\infty}(\mathcal{YZ}_{11}^0),
    $$
    where
    \begin{multline*}
        \mathbb{Z}_{yz}^{F_0} = \sqrt{\pi_1} \mathbb{G} \left[ \Phi_{2}(\Xi(P(X_1)^\top \zeta_{yz}), \Xi(Q(X_1)^\top \zeta_{yz}); \rho_{yz}(R(X_1)^\top \zeta_{yz})) G\right]\\
        +\int_{\mathcal{X}_{11}} \left[ (y - \bar y_y) \Phi^1_2(x,\zeta_{yz})^\top \mathbb{Z}^{\xi}_{\bar y_{y_0}} + (z - \bar z_z) \Phi^2_2(x,\zeta_{yz})^\top
        \mathbb{Z}^{\xi}_{\bar z_{z_0}}\right] \mathrm{d}F_{X \mid T,G}(x \mid 1,1) \\
        +\int_{\mathcal{X}_{11}} \nabla \Phi_2\left(x, \zeta_{yz} \right)^\top  \mathbb{Z}_{\bar y_y \bar z_z}^{\zeta}  \mathrm{d}F_{X \mid T,G}(x \mid 1,1),
    \end{multline*}
   $\pi_1 = \mathrm{P}(G = 1)$, $\Phi^1_2(x,\zeta_{yz}) := \Phi^1_{2}(\Xi(p(x)^\top \zeta_{yz}), \Xi(q(x)^\top \zeta_{yz}); \rho_{yz}(r(x)^\top \zeta_{yz})) p_{\xi}(0)$, $\Phi^2_2(x,\zeta_{yz}) := \Phi^2_{2}(\Xi(p(x)^\top \zeta_{yz}), \Xi(q(x)^\top \zeta_{yz}); \rho_{yz}(r(x)^\top \zeta_{yz})) p_{\xi}(0)$,
   $$
   \nabla \Phi_2(x,\zeta_{yz}) := \left(\begin{array}{ccc}
       \Phi^1_{2}(\Xi(p(x)^\top \zeta_{yz}), \Xi(q(x)^\top \zeta_{yz}); \rho_{yz}(r(x)^\top \zeta_{yz})) p(x) \\
       \Phi^2_{2}(\Xi(p(x)^\top \zeta_{yz}), \Xi(q(x)^\top \zeta_{yz}); \rho_{yz}(r(x)^\top \zeta_{yz})) q(x)\\
       \Phi^3_{2}(\Xi(p(x)^\top \zeta_{yz}), \Xi(q(x)^\top \zeta_{yz}); \rho_{yz}(r(x)^\top \zeta_{yz})) r(x)
   \end{array}\right),    
   $$ 
   and $\Phi_2^j(x_1,x_2,x_3) = \partial_{x_j} \Phi_2(x_1,x_2,x_3),$ $j \in \{1,2,3\}$. 
    (ii) As $n \to \infty$, 
    $$
    \sqrt{n}\left( \hat F^*_{Y^0,Z^0 \mid T,G}(y,z \mid 1,1) - \hat F_{Y^0,Z^0 \mid T,G}(r \mid 1,1)  \right) \rightsquigarrow_{\mathrm{P}} \mathbb{Z}_{yz}^{F_0} \text{ in } \ell^{\infty}(\mathcal{YZ}_{11}^0).
    $$
\end{theorem}


\section{Empirical application}

\label{sec:empirical}

We illustrate our approach through a re-examination of data used in  Card and Krueger (1994), hereafter CK, investigation of the impact on an increase in the minimum wage on the level of employment. In April 1992 New Jersey increased its minimum wage from the Federal level of 4.25 dollars per hour to
5.05 dollars per hour. CK investigated the impact of this increase on the change in the level of  ``full-time equivalent employment", measured as the sum of the number of full-time employees plus half the number of part-time employees, in New Jersey fast food restaurants. They use a DiD estimation strategy in which the control group comprises a group of comparable fast food restaurants in the region of Pennsylvania bordering New Jersey. CK concluded that this particular increase in the minimum wage led to a small increase in the level of full-time equivalent employment. This result produced a large and important related literature on the impact of minimum wages on employment.

We employ our procedure using the CK data to investigate the impact of this increase in the minimum wage on the level of full-time, part-time and full-time equivalent employment respectively.  
We estimate the model using the 409 observations available in CK data set. Of these, 80.9 percent of the observations are from treated establishments in New Jersey. We acknowledge that the data set is relatively small and that this is likely to have implications for the level of statistical significance of the results. However, as our objective is to illustrate our approach in a well-known setting, we prefer to work with a data set which has been frequently used and is well understood (see, for a recent example, Torous et al. 2024) rather than providing our own original application. Note that for the models which include covariates, the additional variables are four dummy variables for franchise type and a dummy variable indicating that the establishment is company owned.

The respective actual and counterfactual distributions are reported in Figure \ref{fig:dist}. These figures are based on the models which include the additional covariates. While there are some differences in the distributions for total employment, indicating an increase in employment for establishment sizes above the 1st quartile, the larger differences appear in Figure \ref{fig:dist}-C. which captures the increases in full-time employment. This figure suggests gains at all quantiles. Figure \ref{fig:dist}-B. is suggestive of some small reductions in part-time employment at some quantiles.

The results of the quantile treatment effects for the univariate analyses are reported in Tables \ref{tab:qte} and \ref{tab:qte_nox} noting that those in the former include the covariates while the latter does not. As the results are generally similar, we focus only on Table 1. The DiD estimate of the mean effect on full-time equivalent employment is 2.65 and this is statistically significant at the 10 percent level. An examination of the table reveals that this mean effect is driven by an increase in full-time employment as there is very weak statistical evidence of a small reduction in mean part-time employment. The level of statistical significance probably reflects the small sample size.

Tables \ref{tab:qte} and \ref{tab:qte_nox} also indicate that the effect of the minimum wage increase is different depending on the size of the establishment. For example, at smaller establishments the effect on full-time equivalent employment is negative and this reflects a reduction in these establishments' levels of part-time employment. However, there appears to be growth in full-time equivalent employment at the larger establishment sizes noting that the level of statistical significance is low.  The point estimates capturing the changes in part-time employment are either zero or negative. The most striking feature of the table is that the increase in full-time equivalent employment is driven by gains in full-time employment. Moreover, the larger gains in employment are at the upper quantiles noting that the estimates with the higher degrees of statistical significance also occur at these quantiles.

As a final exercise, we test the monotonicity of the counterfactual distribution function using the method of Chernozhukov et al. (2010) as presented in their Remark 2. The idea behind their test is that the estimator of the counterfactual distribution $F_{Y^0 \mid G, T}(\cdot \mid 1,1)$ based on rearrangement should have the same asymptotic properties as the one without rearrangement if the model assumptions hold. Therefore, rejecting the null hypothesis of a monotonic $y \mapsto F_{Y^0 \mid G,T}(y \mid 1,1)$ is evidence against the model assumption, in particular against no-interaction. This test can be implementing by verifying if a uniform confidence region for $y \mapsto F_{Y^0 \mid G,T}(y \mid 1,1)$ without rearrangement contains the point estimate with rearrangement. Figure \ref{fig:monotonicity_test} reports this uniform confidence region and shows that the point estimate is within this region. Hence, based on this observation, we cannot reject the null hypothesis that $y \mapsto F_{Y^0 \mid G,T}(y \mid 1,1)$ is monotone, which can be taken as evidence in favor of the no-interaction assumption.

To illustrate the applicability of our methodology to a bivariate analysis we consider the impact of the increase in the minimum wage on the joint distribution of the full-time and part-time employment levels. The counterfactual and actual distributions are shown in Figure \ref{fig:2dim}. The figure reveals that the joint distribution has changed due to the increase in the minimum wage and that the distribution of the treated population appears to have shifted downward and to the right. This appears to be a similar movement to that reported in Figure 5 of Torous et al. (2024). This suggests the treatment has changed the relationship between full-time and part-time employment. However, it is difficult to identify whether this merely reflects the changes in the marginal distributions. It is also difficult to interpret the economic implications of the differences shown in this figure. Accordingly in Table 3 we also present estimates of Kendall's tau and Spearman's correlation index to capture the correlation between these two employment levels. Note that Kendall's tau is calculated as: 
$$
    \hat \tau_d = \frac{2}{n_{11}(n_{11}-1)} \sum_{i=1}^N \sum_{j=i+1}^N G_i G_j T_i T_j\sgn(Y^d_{i} - Y^d_{j}) \sgn(Z^d_{i}-Z^d_{j}), \ \ n_{11} = \sum_{i=1}^N G_i T_i, \ \ d \in \{0,1\},
$$
where $d_i$ indicates whether the Kendall's tau is estimated for the treated or for the untreated sample. Spearman's correlation index is calculated as:

\[
    \hat \rho_{d} = 1 - \frac{6 \sum_{i=1}^{N} G_i T_i R_i^2}{n_{11} (n_{11}^2-1)}, \quad d \in \{0,1\},
\]
where $R_i$ is the difference between the ranks of $Y^d_{i}$ and $Z^d_{i}$ conditional on $G_i=1$ and $T_i=1$ for observation $i$.

The Kendall's tau and the Spearman's correlation index for the treated sample in the second period can be calculated from the observed data. For the counterfactual distribution of the treated sample in the second period when not treated, we first sample from the estimated distribution. That is, we sample a value of $Y^0$ using our estimate of its marginal distribution from above. We then sample $Z^0$ from the conditional distribution of $Z^0 \mid Y^0$ which can be obtained using our estimates for the bivariate model. 



In the absence of treatment the estimates of Kendall's $\tau$ and Spearman's correlation index for these employment levels are -0.0095 and -0.0101 respectively. Following the increase in the minimum wage, this negative relationship becomes stronger, with the corresponding estimate values of -0.1709 and -0.2402. Moreover, despite the relatively small number of observations the difference in the Spearman's correlation is statistically significant at the 10 percent level. These two estimates of the change in the level of correlation both suggest that the increase in the minimum wage has changed the relationship between full-time and part-time employment. The statistically significant stronger negative correlation is consistent with a greater degree of substitutability between part-time and full-time employment in the presence of the higher minimum wage.

\begin{table}
\centering
\begin{turn}{-90}
    \begin{footnotesize}
		\begin{tabular}{lcccccc}
			\hline
			\hline
			& Mean	& 0.1 & 0.25 & 0.5 & 0.75 & 0.9 \\
			\hline
			\hline	
Full-time equivalent employment	&	 2.6554	&	-1.5	&	0.4383	&	2.5	&	1.5	&	1.5\\
95\% confidence intervals 	&	(-0.331,5.642)	&	(-3.528,0.528)	&	(-0.921,1.797)	&	(-0.665,5.665)	&	(-0.94,3.94)	&	(-0.974,3.974)\\
90\% confidence intervals 	&	(0.054,5.257)	&	(-3.48,0.48)	&	(-0.548,1.425)	&	(-0.487,5.487)	&	(-0.539,3.539)	&	(-0.925,3.925)\\
Part-time employment	&	 -0.4618	&	-1.0	&	-2.0	&	-0.5	&	0.0	&	0.0\\
95\% confidence intervals 	&	(-3.596,2.672)	&	(-2.947,0.947)	&	(-4.972,0.972)	&	(-1.963,0.963)	&	(-0.99,0.99)	&	(-0.992,0.992)\\
90\% confidence intervals 	&	(-3.022,2.098)	&	(-2.536,0.536)	&	(-4.946,0.946)	&	(-1.932,0.932)	&	(-0.978,0.978)	&	(-0.982,0.982)\\
Full-time employment	&	 3.0674	&	 $\cdot$	&	2.0	&	1.5	&	3.0	&	5.0\\
95\% confidence intervals 	&	(-0.063,6.197)	&	 $(\cdot, \cdot)$	&	(-0.968,4.968)	&	(-0.956,3.956)	&	(-0.952,6.952)	&	(-0.981,10.981)\\
90\% confidence intervals 	&	(0.49,5.645)	&	 $(\cdot, \cdot)$	&	(-0.923,4.923)	&	(-0.563,3.563)	&	(-0.553,6.553)	&	(-0.084,10.084)\\
\hline \hline
\end{tabular}
\end{footnotesize}
\end{turn}
\caption{Quantile treatment effects}	
\label{tab:qte}
\end{table}

\begin{table}
	\centering
    \begin{turn}{-90}
	\begin{footnotesize}
		\begin{tabular}{lcccccc}
			\hline
			\hline
			& Mean	& 0.1 & 0.25 & 0.5 & 0.75 & 0.9 \\
			\hline
			\hline	
	Full-time equivalent employment	&	 2.565	&	-1.5	&	0.5	&	2.5	&	1.5	&	1.5\\
	95\% confidence intervals 	&	(-0.785,5.915)	&	(-3.536,0.536)	&	(-0.928,1.928)	&	(-0.663,5.663)	&	(-0.726,3.726)	&	(-0.982,3.982)\\
	90\% confidence intervals 	&	(-0.371,5.501)	&	(-3.486,0.486)	&	(-0.557,1.557)	&	(-0.518,5.518)	&	(-0.51,3.51)	&	(-0.938,3.938)\\
	Part-time employment	&	 -0.5134	&	-1.0	&	-2.0	&	1.0	&	0.0	&	0.0\\
	95\% confidence intervals 	&	(-3.926,2.899)	&	(-2.511,0.511)	&	(-4.971,0.971)	&	(-0.977,2.977)	&	(-0.991,0.991)	&	(-0.992,0.992)\\
	90\% confidence intervals 	&	(-3.463,2.436)	&	(-2.447,0.447)	&	(-4.949,0.949)	&	(-0.944,2.944)	&	(-0.976,0.976)	&	(-0.982,0.982)\\
	Full-time employment	&	 3.0512	&	 $\cdot$	&	2.0	&	2.0	&	3.0	&	5.0\\
	95\% confidence intervals 	&	(-0.254,6.357)	&	 $(\cdot, \cdot)$	&	(-0.965,4.965)	&	(-0.929,4.929)	&	(-0.968,6.968)	&	(-0.976,10.976)\\
	90\% confidence intervals 	&	(0.267,5.835)	&	 $(\cdot, \cdot)$	&	(-0.934,4.934)	&	(-0.461,4.461)	&	(-0.951,6.951)	&	(-0.956,10.956)\\
	\hline\hline
\end{tabular}
\end{footnotesize}
\end{turn}
\caption{Quantile treatment effects -- No additional covariates included}	
\label{tab:qte_nox}
\end{table}

\begin{figure}
\centering
\resizebox{1.1\textwidth}{!}{
\input{card_krueger/figure_results_2006_original.tex}}
\caption{Actual and counterfactual distributions for total, parttime, and full-time employment.}
\label{fig:dist}
\end{figure}

\begin{figure}
\centering
\resizebox{1.1\textwidth}{!}{
\input{card_krueger/figure_results_check.tex}}
\caption{Confidence intervals for the estimator without rearrangements and point-estimators using rearrangements.}
\label{fig:monotonicity_test}
\end{figure}

\begin{figure}
	\caption{Joint Distribution of Full-time and Part-time Employment with and without Treatment}
    \label{fig:2dim}
	\input{card_krueger/figure_bivariate_wages_unemployment}
	\input{card_krueger/figure_bivariate_wages_unemployment1}
\end{figure}

\begin{table}
\caption{Estimates of the Kendall's $\tau$ and Spearman's correlation index.}
\centering
\begin{turn}{-90}
\begin{tabular}{lccc}
	\hline
	\hline
	&	\multicolumn{1}{c}{Treatment} & \multicolumn{1}{c}{Without treatment}  & \multicolumn{1}{c}{Difference}\\
	\hline
	\hline
	\multicolumn{4}{l}{Kendall's $\tau$}\\
	\hline
	Part-time and full-time employment 	&-0.1709	&	-0.0095	&	-0.1613\\
	95\% confidence intervals 	&	(-0.2407,-0.1011)	&	(-0.2401,0.221)	&	(-0.3911,0.0684)\\
	95\% confidence intervals 	&	(-0.2305,-0.1113)	&	(-0.1848,0.1657)	&	(-0.339,0.0164)\\
	\hline
	\multicolumn{4}{l}{Spearman's correlation index}\\
	\hline
	Part-time and full-time employment 	&-0.2402	&	-0.0101	&	-0.23\\
	95\% confidence intervals 	&	(-0.3383,-0.142)	&	(-0.2596,0.2393)	&	(-0.5027,0.0325)\\
	90\% confidence intervals 	&	(-0.3216,-0.1588)	&	(-0.1923,0.1721)	&	(-0.4192,-0.0409)\\
	\hline
	\hline
\end{tabular}
\end{turn}
\end{table}

\begin{table}
	\caption{Results of the Kendall's $\tau$ and Spearman's correlation index -- No additional covariates.}
    \centering
    \begin{turn}{-90}
\begin{tabular}{lccc}
	\hline
	\hline
	&	\multicolumn{1}{c}{Treatment} & \multicolumn{1}{c}{Without treatment}  & \multicolumn{1}{c}{Difference}\\
	\hline
	\hline
	\multicolumn{4}{l}{Kendall's $\tau$}\\
	\hline
	Part-time and full-time employment 	&-0.1709	&	-0.0136	&	-0.1573\\
	95\% confidence intervals 	&	(-0.23,-0.1117)	&	(-0.2127,0.1855)&	(-0.371,0.0565)\\
	95\% confidence intervals 	&	(-0.2268,-0.1149)	&	(-0.2035,0.1764)	&	(-0.3374,0.0228)\\
	\hline
	\multicolumn{4}{l}{Spearman's correlation index}\\
	\hline
	Part-time and full-time employment 	&-0.2402	&	-0.0143	&	-0.2259\\
	95\% confidence intervals 	&	(-0.3237,-0.1566)	&	(-0.2352,0.2066)	&	(-0.4751,0.0091)\\
	90\% confidence intervals 	&	(-0.3172,-0.1632)	&	(-0.2148,0.1862)	&	(-0.4316,-0.0202)\\
	\hline
	\hline
\end{tabular}
\end{turn}
\end{table}

\section{Conclusion}

\label{sec:conclusions}
We provide a simple distribution regression based estimator to implement the evaluation of treatment effects in a difference-in-difference setting. As our approach provides counterfactual distributions we are able to explore the impact of the treatment at different quantiles of the distribution of the outcome variable. For both the univariate and multivariate cases we provide the identifying assumption and the associated estimation algorithms. A re-examination of the Card and Krueger (1994) study highlights the utility of various aspects of our approach.

Our analysis can easily be extended to the case of multiple time periods and more than two outcomes. We can also extend our distributional regression framework to use time and unit weights as in the synthetic difference-in-difference estimation method of Arkhangelsky et al. (2021). We leave each of these extensions to future research (e.g. Fern\'andez-Val et al., 2024).

\section*{\thinspace\ References}


\begin{description}

\item \textsc{Almond, D., H.W. Hoynes, and D.W. Schanzenbach} (2011), \textquotedblleft Inside the war on poverty: the impact of food stamps
on birth outcomes\textquotedblright, \emph{Review of Economics and Statistics} \textbf{93}, 387--403.

\item \textsc{Arkhangelsky, D. and G. Imbens} (2024), \textquotedblleft Causal Models for longitudinal and
panel Data: a survey\textquotedblright, \emph{The Econometrics Journal} \textbf{27}, C1-61.

\item \textsc{Arkhangelsky, D. S. Athey, D.A. Hirshberg, G.W. Imbens, and S. Wager} (2021), \textquotedblleft Synthetic difference-in-differences, \emph{American Economic Review} \textbf{111}, 4088–18.

\item \textsc{Athey, S. and G.D. Imbens} (2006),
\textquotedblleft Identification and inference in nonlinear
difference-in-differences models\textquotedblright, \emph{Econometrica} 
\textbf{74}, 431--97.

\item \textsc{Biewen, M., M. R\"ummele, and B. Fitzenberger} (2022), \textquotedblleft Using distribution regression Difference-in-Differences to evaluate the effects of a minimum wage introduction on the distribution of hourly wages and hours worked, working paper, N\"urnberg.

\item \textsc{Blundell, R., C. Meghir, M. Costa Dias and J. Van Reenen} (2004), \textquotedblleft Evaluating the employment impact of a mandatory job search program\textquotedblright, \emph{Journal of the European Economic Association} \textbf{2}, 569--606.

\item \textsc{Callaway,  B. and T. Li} (2019), \textquotedblleft Quantile treatment effects in difference in differences models with panel data\textquotedblright, \emph{Quantitative Economics} \textbf{10}, 1579--1618.

\item \textsc{Card, D.} (1990),  \textquotedblleft The Impact of the Mariel Boatlift on the Miami Labor Market\textquotedblright, \emph{Industrial and Labor Relations Review}, \textbf{43}, 245–-57.

\item \textsc{Card, D. and A.B. Krueger} (1994),  \textquotedblleft Minimum Wages and Employment: A Case Study of the Fast-Food Industry in New Jersey and Pennsylvania\textquotedblright \emph{American Economic Review} \textbf{84}, 772–-93.

\item \textsc{Cengiz, D., A. Dube, A. Lindner, and B. Zipperer} (2019), \textquotedblleft The effect of minimum wages on low-wage jobs \textquotedblright, \emph{Quarterly Journal of Economics} \textbf{134}, 1405--54.

\item \textsc{Chernozhukov, V. I. Fern\'andez-Val and A. Galichon} (2010), \textquotedblleft Quantile and probability curves without crossing \textquotedblright, \emph{Econometrica} \textbf{78}, 1093--1125.

\item \textsc{Chernozhukov, V., I. Fern\'{a}ndez-Val, and Melly B.}
(2013), \textquotedblleft Inference on counterfactual
distributions\textquotedblright , \emph{Econometrica}, \textbf{81}, 2205--68.

\item \textsc{Chernozhukov, V., I. Fern\'{a}ndez-Val, and Melly B.}
(2022), \textquotedblleft Fast algorithms for the quantile regression process\textquotedblright , \emph{Empirical economics}, \textbf{62(1)}, 7--33.

\item \textsc{Chernozhukov, V., I. Fern\'{a}ndez-Val, and S. Luo}{\small \ }%
(2019), \textquotedblleft Distribution regression with sample selection,
with an application to wage decompositions in the UK", working paper, MIT,
Cambridge (MA).

\item \textsc{Chernozhukov, V. I. Fern\'andez-Val, B. Melly, and K. Wüthrich} (2020), \textquotedblleft Generic inference on quantile and quantile effect functions for discrete outcomes\textquotedblright, \emph{Journal of the American Statistical Association} \textbf{115}, 123--37.

\item \textsc{Chernozhukov, V., Fern\'{a}ndez-Val I., J. Meier, A. van Vuuren and F. Vella}
(2024) \textquotedblleft Conditional Rank-Rank Regression \textquotedblright, arXiv preprint arXiv:2407.06387.

\item \textsc{Chernozhukov, V., Fern\'{a}ndez-Val I., J. Meier, A. van Vuuren and F. Vella}
(2025) \textquotedblleft Bivariate distribution regression with an application to
intergenerational mobility\textquotedblright, arXiv preprint	arXiv:2508.12716.

\item \textsc{Dube, A.} (2019), \textquotedblleft Minimum wages and the distribution of family incomes\textquotedblright, \emph{American Economic Journal: Applied Economics} \textbf{11}, 268–-304.


\item \textsc{Fern\'{a}ndez-Val I., J. Meier, A. van Vuuren and F. Vella}
(2024) \textquotedblleft Distributional synthetic difference-in-differences\textquotedblright, working paper, Boston University.

\item \textsc{Foresi, S. and F. Peracchi} (1995), ``The conditional
distribution of excess returns: an empirical analysis'', \emph{Journal of
the American Statistical Association}, \textbf{90}, 451--66.

\item \textsc{Goodman-Bacon, A.} (2021), \textquotedblleft The long-run effects of childhood insurance coverage: medicaid implementation, adult health, and labor market outcomes, \emph{American Economic Review} \textbf{111}, 2550-93.

\item \textsc{Goodman-Bacon, A. and L. Schmidt} (2020), \textquotedblleft Federalizing benefits: The introduction of supplemental security income and the size of the safety net.\textquotedblright, \emph{Journal of Public Economics} \textbf{185}, 104174.

\item \textsc{Kim, D. and J.M. Wooldridge} (2024), \textquotedblleft Difference-in-differences estimator of quantile
treatment effect on the treated \textquotedblright, \emph{Journal of Business and Economic Statistics} \textbf{43}, 401--12.

\item \textsc{MaCurdy, T.} (2015), \textquotedblleft How effective is the minimum wage at supporting the poor?\textquotedblright, \emph{Journal of Political Economy} \textbf{123}, 497--545.

\item \textsc{Malesky, E.J., C.V. Nguyen, and A. Trahn} (2014),
\textquotedblleft The Impact of recentralization on public services:
A difference-in-differences analysis of the abolition
of elected councils in Vietnam\textquotedblright, \emph{American Political Science Review} 
\textbf{108}, 144--68.

\item \textsc{Melly, B. and Santangelo} (2015), \textquotedblleft The changes-in-changes model with covariates\textquotedblright, working paper, Bern University.

\item \textsc{Roth, J. and P.H.C. Sant'Anna} (2023),
\textquotedblleft When Is parallel trends sensitive to functional form?\textquotedblright, \emph{Econometrica} 
\textbf{91}, 737--47.

\item \textsc{Snow, J.} (1855), \textquotedblleft \emph{On the Mode of Communication of Cholera.}\textquotedblright 2nd ed. John Churchill.

\item \textsc{Torous, W., F. Gunsilius, and P. Rigollet} (2024), \textquotedblleft An optimal transport approach to estimating
causal effects via nonlinear difference-in-differences, working paper\textquotedblright, University of California, Berkeley.

\item \textsc{Van der Vaart, A.W.} (1998), \textquotedblleft \emph{Asymptotic statistics} \textquotedblright, Cambridge University Press.

\item \textsc{Van der Vaart, A.W. and J.A. Wellner} (2013), \textquotedblleft  Weak Convergence and empirical processes with applications to statistics \textquotedblright, Springer, Cham.

\item \textsc{Wooldridge, J.M.} (2023),
\textquotedblleft Simple approaches to nonlinear difference-in-differences with panel data\textquotedblright, \emph{Econometrics Journal} 
\textbf{26}, C31--66.

\item \textsc{Williams, O. D. and J.E. Grizzle} (1972), \textquotedblleft Analysis of contingency tables having ordered response categories\textquotedblright, \emph{Journal of the American Statistical Association} \textbf{67}, 55–63.

\end{description}

\appendix
\renewcommand{\baselinestretch}{1}

\section{Proofs}

\subsection{Proof of Lemma \ref{lemma:asymptotic_univariate}}
\label{app:asymptotic_univariate}
The result follows from the same steps as in Lemma 6.1 of Chernozhukov et al. (2024) after replacing the expression of the conditional log-likelihood function in the cross-sectional case by the expression of the average partial log-likelihood in the panel case that aggregates over the  two time periods in \eqref{eq:plik-univ}.

\subsection{Proof of Lemma \ref{lemma:asymptotic_bivariate}}\label{appendix:asymptotic_bivariate}
The result follows from the same steps as in Lemma 5.1 of Chernozhukov et al. (2025) after replacing the expression of the conditional log-likelihood function in equation (4.2) of Chernozhukov et al. (2025) by the expression of the average partial log-likelihood that aggregates over the  two time periods in \eqref{eq:plik-biv}.

\end{document}